\documentclass[12pt]{article}      

\usepackage{graphicx}
\usepackage{amssymb}
\usepackage{amsmath}
\usepackage{textcomp}
\usepackage{float}
\usepackage{fancyvrb}
\usepackage{caption}
\usepackage{graphicx,amsmath,amsfonts,amsthm,bbm,setspace,url}
\usepackage{verbatim,mathrsfs}
\usepackage{natbib}

\usepackage{mymacros}

\DeclareMathOperator*{\argmin}{arg\,min}

\advance\voffset by -0.25cm
\begin{document}
\baselineskip=20pt

\begin{center}{
  \Large \bf
  Penalized basis models for very large spatial datasets
}

\bigskip
\bigskip

{\bf 
  Mitchell Krock\footnote[1]{Department of Applied Mathematics,
  University of Colorado, Boulder, CO.  
  Corresponding author e-mail:
  \texttt{mitchell.krock@colorado.edu}}, 
  William Kleiber\footnotemark[1] 
  and 
  Stephen Becker\footnotemark[1]
}

\bigskip
\bigskip

{\bf \today}

\end{center}

\bigskip
\bigskip

\begin{abstract}
Many modern spatial models express the stochastic variation component as a 
basis expansion with random coefficients.  Low rank models, approximate spectral 
decompositions, multiresolution representations, stochastic partial differential 
equations and empirical orthogonal functions all fall within this basic framework. 
Given a particular basis, stochastic dependence relies on flexible modeling 
of the coefficients. Under a Gaussianity assumption, we propose a 
graphical model family for the stochastic coefficients by parameterizing the 
precision matrix.  Sparsity in the precision matrix is encouraged using a 
penalized likelihood framework.  Computations follow from a majorization-minimization 
approach, a byproduct of which is a connection to the graphical lasso. 
The result is a flexible nonstationary spatial model that is adaptable to 
very large datasets. We apply the model to two large and heterogeneous 
spatial datasets in statistical climatology and recover physically sensible 
graphical structures. Moreover, the model performs competitively against 
the popular LatticeKrig model in predictive cross-validation, but substantially 
improves the Akaike information criterion score.
\end{abstract}


\section{Introduction}

Many modern spatial models express the stochastic variation component as a 
basis expansion with random coefficients.  Low rank models, approximate spectral 
decompositions, multiresolution representations, stochastic partial differential 
equations and empirical orthogonal functions all fall within this basic framework. 
The essential difference between these methods is the amount of modeling effort 
placed on the basis versus the coefficients. 

We introduce a novel approach applicable to any model within this framework 
that allows for nonstationarity and easily adapts to large datasets using 
off-the-shelf popular basis models. 
The method allows for straightforward graphical interpretations of the 
conditional independence structure of the stochastic coefficients. 

Most spatial statistical models for an observational process $Y(\bs)$ with 
$\bs\in\real^d$ can be written 
\begin{equation}  \label{eq:obs.model}
  Y(\bs) = \mu(\bs) + Z(\bs) + \vep(\bs)
\end{equation}
which decomposes the observations into a constant mean function $\mu$, 
a spatially-correlated random deviation $Z$ and a white noise process $\vep$. 
Flexible models for the correlated deviation $Z$ are necessary, and much research 
has been devoted to exploring classes of practical specifications. 

In this work we focus on a particularly popular framework where 
\begin{equation}  \label{eq:basis}
  Z(\bs) = \sum_{j=1}^\ell c_j \phi_j(\bs)
\end{equation}
for $\ell$ fixed basis functions $\phi_1,\dots,\phi_\ell$ 
and stochastic coefficients $\bc = (c_1,\ldots,c_\ell)\T \sim N(0,Q^{-1})$. 
Such a spatial basis expansion subsumes many popular approaches including 
discretized frequency domain models \citep{fuentes2002, matsuda2009, bandy2010}, 
empirical orthogonal functions \citep[Ch.\ 5,][]{wikle}, 
low rank representations \citep{cressie2008,banerjee2008,guhaniyogi2017}, 
multiresolution and wavelet representations \citep{nychka2002,nychka2015,katzfuss}, 
and stochastic partial differential equation models \citep{lindgren2011,bolin}.

To set the stage, let us briefly describe some specific models of (\ref{eq:basis}). 
Fixed rank kriging sets $\ell$ to be relatively small compared to the sample size, 
uses bisquare functions for $\phi_i$, and $Q^{-1}$ is not modeled but rather 
estimated by minimizing a squared Frobenius distance from an binned empirical 
covariance matrix \citep{cressie2008}.  
The more recent LatticeKrig model is a multiresolution model that places a large 
number of compactly supported basis functions with varying supports on a grid 
and specifies $\bc$ as a Gaussian Markov random field \citep{nychka2015}. 
Discretized frequency domain approaches and empirical orthogonal functions set the 
basis functions to be globally supported with independent coefficients (diagonal 
$Q$), mirroring the spectral representation theorem or the Karhunen-Loeve expansion 
for stochastic processes. 
In this work we attempt to relax the modeling assumptions on the structure 
of the stochastic coefficients $\bc$ (equivalently $Q$) by assuming only that 
they arise from a Gaussian graphical model.  We use the well known fact that 
the graph structure of a multivariate Gaussian is equivalent to the 
zero/non-zero pattern in the precision matrix $Q$ \citep{rue2005}. 
A major distinction of this work is that we 
do not specify the structure of the graph, but instead try to infer its 
edges by estimating the entries of $Q$, while encouraging sparsity using a 
penalized likelihood estimation framework. 

The penalization framework appears straightforward, using an $\ell_1$-penalized 
maximum likelihood estimator for $Q$. However, the corresponding optimization 
problem is nonconvex, but we show it can be solved efficiently using a 
majorization-minimization approach where the interior minimization problem 
corresponds to the graphical lasso problem \citep{friedman_sparse_2008}. 
We perform a relatively detailed simulation study to assess the algorithm and 
model's ability to recover unknown graphical structures, and also apply the 
method to two challenging large and heterogeneous datasets: the first 
a historical observational dataset of minimum temperatures over a portion of 
North America, and the second a global reforecast dataset of surface temperature. 
The results from these real data applications suggest our method can 
appropriately capture nonstationary spatial correlations with minimal modeling 
effort. 
 
\section{Penalized likelihood estimation}

For ease of exposition we suppose $\mu(\bs) = 0$ in (\ref{eq:obs.model}).  
Thus the observational model is 
\begin{equation}  \label{basis}
  Y(\bs) = \sum_{j=1}^\ell c_j \phi_j(\bs)+ \vep(\bs).
\end{equation}
We suppose we have $m$ independent realizations, $Y(\bs)=Y_1(\bs),\ldots,Y_m(\bs)$ 
of the observational process at spatial locations $\bs=\bs_1,\ldots,\bs_n$.  
Group a realization as $\bY_i = (Y_i(\bs_1),\ldots,Y_i(\bs_n))\T$.  A matrix 
representation of the model is 
\begin{equation} \label{lineargaussianmodel}
  \bY_i = \Phi \bc_i + \bvep_i, \quad i=1,\dots,m
\end{equation}
where $\Phi$ is an $n \times \ell$ matrix with $(i,j)$th entry $\phi_j(\bs_i)$, 
$\bc_i = (c_{i1},\ldots,c_{i\ell})\T$ are $m$ independent vectors of the 
stochastic coefficients and $\bvep_i = (\vep_i(\bs_1),\ldots,\vep_i(\bs_n))\T$ 
are $m$ independent realizations of the white noise process. 
The stochastic assumptions of our model are that $\vep_i$ is a mean 
zero white noise process with variance $\tau^2>0$, commonly referred to as the 
nugget effect in the geostatistical literature, and that $\bc_i$ is a mean zero 
$\ell$-variate multivariate normal random vector with precision matrix $Q$. 
The zero structure of $Q$ encodes the graphical model for $\bc_i$.

The model (\ref{basis}) plays a crucial role in modern statistics: it is the 
framework for a variety of popular statistical techniques including factor 
analysis, principal component analysis, linear dynamical systems, hidden 
Markov models, and relevance vector machines \citep{unifying, rvm}. In the spatial 
context, there are two main features that arise from using a model of the form 
(\ref{basis}): first, the resulting model of the spatial field is nonstationary, 
and second, common computations involving the covariance matrix can be 
sped up with particular choices for $\Phi$ and $Q$.


\subsection{The penalized likelihood}

With the $m$ realizations $\bY_1,\ldots,\bY_m$ of (\ref{lineargaussianmodel}), 
twice the negative log-likelihood function can be written, up to multiplicative and 
additive constants, as
\begin{equation}\label{neg.log.lik}
  \log \det  ( \Phi Q^{-1} \Phi^{\mathrm{T}} + \tau^2 I_n) 
  + \text{tr}(S (\Phi Q^{-1} \Phi^{\mathrm{T}} + \tau^2 I_n)^{-1})
\end{equation}
where $S = \frac{1}{m} \sum_{i=1}^m\bY_i \bY_i\T$ is an empirical covariance matrix. 

Our goal is to estimate $Q$ under the assumption that it follows a graphical 
structure. A few connections are worth noting.  When $\Phi = I_\ell$ and $\tau^2=0$, 
(\ref{neg.log.lik}) is, up to the regularization term, the graphical lasso problem 
studied in \citep{friedman_sparse_2008}. 
In particular, this is equivalent to directly observing $\bc_1,\dots,\bc_m$.  
The graphical lasso uses an $\ell_1$ penalty to induce a graph structure 
on $Q$, from which we draw inspiration next. 
Our situation is substantially more complicated due to observational error and 
indirect observations of $\bc_i$ that are modulated by $\Phi$. 

Our proposal is to estimate $Q$ by minimizing a penalized version of 
(\ref{neg.log.lik}).  The estimator of $Q$ is 
\begin{equation}  \label{basisMLE}
 \hat{Q} \in \argmin \limits_{Q \succeq 0}  \  
  \log \det( \Phi Q^{-1} \Phi^{\mathrm{T}} + \tau^2 I_n) 
  + \text{tr}(S (\Phi Q^{-1} \Phi^{\mathrm{T}} + \tau^2 I_n)^{-1}) 
  + \|\Lambda \circ Q\|_1.
\end{equation}
The notation $Q \succeq 0$ indicates that $Q$ must be positive semidefinite, and 
$\|\Lambda \circ Q\|_1 = \sum_{i,j} \Lambda_{ij} |Q_{ij}|$ is a penalty term that enforces sparsity on the elements of $Q$. Here $\Lambda_{ij}$ are nonnegative penalty parameters, 
with higher values in the matrix $\Lambda$ encouraging more zeros 
in the estimate. In this paper, we assume that the diagonal elements of 
$\Lambda$ are zero, reflecting the fact that we are searching for sparsity 
in the off-diagonal elements of $Q$.

At first glance it is hard to determine whether the problem (\ref{basisMLE}) is 
convex or nonconvex. We address this issue explicitly in the next section.
In either case, it will be difficult to work with the 
objective function presented in (\ref{basisMLE}) due to the nested inverses 
surrounding $Q$. Employing a slew of matrix identities allows for the following 
result, whose proof is kept in the Appendix. 
\begin{proposition}  \label{prop:likelihood}
  The minimizer of (\ref{basisMLE}) is also the minimizer of 
  \begin{equation}  \label{Qlikelihood}
    \log \det \left(Q + \frac{ 1}{ \tau^2}\Phi^{\mathrm{T}} \Phi \right) 
    - \log \det Q - \mathrm{tr} \left(  \dfrac{1}{\tau^4} \Phi^{\mathrm{T}} S \Phi 
    \left( Q + \frac{1}{\tau^2} \Phi^{\mathrm{T}} \Phi  \right) ^{-1} \right) 
    + \|\Lambda \circ Q\|_1.
  \end{equation}
\end{proposition}

After precomputing the matrix products $\Phi^{\mathrm{T}} \Phi$ and 
$\Phi^{\mathrm{T}} S \Phi$, evaluating (\ref{Qlikelihood}) involves inverses 
and determinants of only $\ell \times \ell$ matrices. This is the essential 
computational strategy of fixed rank kriging and LatticeKrig which is harnessed in different ways: by making either 
$\ell$ noticeably smaller than $n$ (fixed rank kriging), or by making $\ell$ extremely large but ensuring the resulting matrices are 
sparse (LatticeKrig).  

\subsection{Optimization approach} \label{algorithmsection}


In the $\ell=1$ case, (\ref{Qlikelihood}) is a univariate function that is twice differentiable on the positive real line. It is straightforward to select $\Phi^{\mathrm{T}} \Phi$, $\Phi^{\mathrm{T}}S\Phi$, and $\tau^2$ so that the second derivative has a negative value at some point along the positive real line. Thus (\ref{Qlikelihood}) is, in general, a nonconvex function on $Q \succeq 0$.
We can, however, show that the four summands in
(\ref{Qlikelihood}) are concave, convex, concave, and convex,
respectively, on $Q \succeq 0$; see the Appendix for details. 
Therefore, the objective function in (\ref{Qlikelihood}) can be written as
\begin{equation} \label{symboliclikelihood}
  \argmin \limits_{Q \succeq 0} f(Q) + g(Q) + \| \Lambda \circ Q\|_1
\end{equation}
where $f(Q) + \| \Lambda \circ Q\|_1$ is convex and $g(Q)$ is concave and differentiable.
A natural approach for this nonconvex problem is a Difference-of-Convex (DC) program \citep{dc} where we iteratively linearize the concave part $g(Q)$ at the previous guess $Q_j$ and solve the resulting convex problem:
\begin{equation} \label{DC}
  Q_{j+1} = \argmin \limits_{Q \succeq 0} \left( f(Q) + \text{tr}(\nabla g(Q_j)  Q) + 
  \|\Lambda \circ Q\|_1\right).
\end{equation}
Motivating the DC framework requires describing a more general scheme called a {\it majorization-minimization algorithm} \citep{mm}. We say that a function $h(\theta)$ is majorized by $m(\theta \mid \theta^*)$ at $\theta^*$ if $h(\theta) \le m(\theta \mid \theta^*)$ for all $\theta$ and $h(\theta^*) = m(\theta^* \mid \theta^*)$. Instead of directly minimizing $h(\theta)$, which can be very complicated, a majorizaton-minimization (MM) algorithm solves a sequence of minimization problems where the majorizing function at the previous guess is minimized:
\begin{equation} \label{mmalgorithm}
  \theta_{j+1} \leftarrow \argmin \limits_\theta  \ m(\theta \mid \theta_j).
\end{equation}
Combining (\ref{mmalgorithm}) with the definition of a majorant yields the inequality
\[
  h(\theta_{j+1}) \le m(\theta_{j+1} \mid \theta_{j}) \le 
  m(\theta_j \mid \theta_{j}) = h(\theta_j)
\]
and thus the algorithm is forced to a local minimum (or saddle point) of $h(\theta)$. The most famous instance of the MM algorithm in statistics is the expectation maximization (EM) algorithm, which under this framework uses Jensen's Inequality to construct majorizing functions for the conditional expectation of log likelihood equations.

Difference-of-convex programming, also called the Concave-Convex-Procedure (CCP), is a subclass of MM where the supporting hyperplane inequality $g(\theta) \le g(\theta_j) +\langle \nabla g(\theta_j), \theta-\theta_j \rangle$ is used to construct a majorizing function when $h(\theta)$ is written as the sum of a convex function and a concave differentiable function $g(\theta)$; i.e.,\ when $h(\theta)$ is a difference of convex functions. An added benefit under the DC framework is that the majorizing function is convex by construction and hence we solve a series of convex optimization problems in each step of (\ref{mmalgorithm}).

In our likelihood function, the convex part is
\[
  f(Q) = - \log \det Q
\]
and the concave part is
\[
  g(Q) = \log \det \left(Q + \frac{1}{ \tau^2} \Phi^{\mathrm{T}} \Phi  \right) - 
  \text{tr} \left(  \dfrac{1}{\tau^4} \Phi^{\mathrm{T}}  S \Phi 
  \left( Q + \frac{1}{\tau^2} \Phi^{\mathrm{T}} \Phi  \right) ^{-1} \right), 
\]
so the DC algorithm (\ref{DC}) becomes
\begin{equation}\label{QUICproblem}
  Q_{j+1} = \argmin_{Q \succeq 0}   \left(  - \log \det Q+\text{tr} 
  \left(  \nabla g(Q_j) Q\right) + \|\Lambda \circ Q\|_1 \right)
\end{equation}
where $\nabla g(Q_j)= \left(  I_\ell + \left ( Q_j + \frac{1}{\tau^2} \Phi^{\mathrm{T}} \Phi \right)^{-1}  \frac{1}{\tau^4} \Phi^{\mathrm{T}} S \Phi\right) \left( Q_j + \frac{1}{\tau^2} \Phi^{\mathrm{T}} \Phi \right)^{-1}$. The inner minimization problem in (\ref{QUICproblem}) is well-studied and known in statistics as the {\it graphical lasso} problem. The graphical lasso is used to estimate an undirected graphical model (i.e.\ GMRF neighborhood structure) for a mean zero multivariate Gaussian vector $\mathbf{X}$ under the assumption that we observe $\mathbf{X}$ directly and without noise. The standard graphical lasso estimate is obtained from the penalized negative log likelihood
\begin{equation} \label{graphicallasso}
  \argmin \limits_{Q \succeq 0} \ - \log \det Q  + \text{tr}(S_{\mathbf{X}}Q) 
  +  \|\Lambda \circ Q\|_1
\end{equation}
where $S_{\mathbf{X}}$ is the sample covariance of the realizations of $\mathbf{X}$. In effect, we have shown that the graphical structure of $Q$ given realizations from $\Phi \mathbf{c} + \boldsymbol \varepsilon$ can be discerned through a Concave-Convex-Procedure where the inner solve is a graphical lasso problem with the ``sample covariance'' matrix as a function of the previous guess $Q_j$.

A variety of numerical techniques have been proposed for the graphical lasso. \citep{yuanlin}  and \citep{banerjee_model_2008} both use interior point methods, but the latter examine the {\it dual problem} of (\ref{graphicallasso})
\begin{equation} \label{dual}
  \argmin \limits_{\|U\|_\infty \le \lambda} \ - \log \det (S+U) - n,
\end{equation}
where $\|U\|_\infty$ is the maximum absolute element of the matrix $U$, and solve (\ref{dual}) one column at a time via quadratic programming. \citep{friedman_sparse_2008} takes an identical approach but writes the dual problem of the columnwise minimization as a lasso regression, which they solve quickly using their own coordinate descent algorithm \citep{pathwise}. This implementation is available in the popular \texttt{R} package \texttt{glasso}. 
 
Advances in solving (\ref{graphicallasso}) in recent years have stemmed from the use of second-order methods that incorporate Hessian information instead of simply the gradient. A current state of the art algorithm is \texttt{QUIC} \citep{QUIC} which also features an \texttt{R} package of the same name. Briefly, the \texttt{QUIC} algorithm uses coordinate descent to search for a Newton direction based on a quadratic expansion about the previous guess and then an Armijo rule to select the corresponding stepsize. During the coordinate descent update, only a set of free variables are updated, making the procedure particularly effective when $Q$ is sparse. In the appendix, we provide \texttt{R} Code to solve (\ref{QUICproblem}) using \texttt{QUIC}. A very recent paper \citep{lineartime2019}, accompanied by Matlab code, shows that reformulating (\ref{graphicallasso}) as a maximum determinant matrix completion problem is a promising strategy.

\subsubsection{Estimating the nugget variance} \label{nug_estimate}

In practice, we must produce an estimate $\hat{\tau}^2$ which is fixed during the algorithm (\ref{QUICproblem}). For this purpose we return to the likelihood (\ref{neg.log.lik}), now rewritten as
\[
  f(Q,\tau^2) =  \log \det \left(Q + \frac{ 1}{\tau^2} \Phi^T \Phi  \right) 
  - \log \det Q  -  \text{tr} \left( \frac{1}{\tau^4} \Phi^T S \Phi 
  \left( Q + \frac{1}{\tau^2} \Phi^T \Phi \right) ^{-1}  \right) 
  + n \log \tau^2 + \frac{\text{tr}(S)}{\tau^2},
\]
under the assumption that $Q = \alpha I_\ell$ for some parameter $\alpha>0$. 
The maximizer of $f$ over $\tau^2$ and $\alpha$ yields an estimate for the nugget effect of our nonstationary process by an approximation to stationarity through $Q = \alpha I_\ell$. Our estimates $\hat{\tau}^2$ and $\hat{\alpha}$ were retrieved from an L-BFGS optimization routine via the \texttt{optim} function in \texttt{R}. In the simulation study below this approach is seen to empirically work very well. Jointly estimating a full model of $Q$ and $\tau^2$ is complicated and unlikely to result in substantial empirical improvement (see Section \ref{sim_comments}).

\subsubsection{Estimating the penalty weight} \label{lambda_estimate}

All that remains to specify is the penalty weight matrix $\Lambda$. One option follows \cite{bien}, in which a likelihood-based cross validation approach is used to select $\Lambda$ in the context of estimating a sparse covariance matrix. More formally, suppose we use $k$ folds and consider $t$ penalty matrices $(\Lambda_1,\dots,\Lambda_t)$. Let $\hat Q_\Lambda(S)$ be the estimate we get from applying our algorithm with empirical covariance $S = \frac{1}{m} \sum_{i=1}^m \bY_i\bY_i\T$ and penalty $\Lambda$. For $A \subseteq \{1,\dots,m\}$, let $S_A = |A|^{-1} \sum_{i \in A} \bY_i\bY_i\T$. We seek $\Lambda$ so that $\alpha(\Lambda)= \ell(\hat Q_\Lambda(S), S)$ is small, where 
\begin{equation} \label{finallikelihood}
  \ell(Q,S) = \log \det \left(Q + \frac{ 1}{\tau^2} \Phi^T \Phi  \right) 
  - \log \det Q -  \text{tr} \left( \dfrac{1}{\tau^4} \Phi^{\mathrm{T}}  S \Phi 
  \left( Q + \frac{1}{\tau^2} \Phi^{\mathrm{T}} \Phi  \right) ^{-1} \right) 
\end{equation}
is the unpenalized likelihood function in (\ref{Qlikelihood}). The cross validation approach is to partition $\{1,\dots,m\}$ into disjoint sets $\{A_1,\dots,A_k\}$ and select $\hat \Lambda = \argmin \limits_{ \Lambda \in \{\Lambda_1,\dots,\Lambda_t \}} \hat \alpha(\Lambda)$ where  $\hat \alpha(\Lambda) = k^{-1} \sum_{i=1}^k \ell(\hat Q_\Lambda (S_{A_i^c}),S_{A_i} )$.

Another option is through spatial cross validation, where we use training data to estimate estimate $\hat{Q}_{\Lambda_1}(S_{\text{train}}),\dots,\hat{Q}_{\Lambda_t}(S_{\text{train}})$ and then krige to the held out locations in the testing data. The penalty weight matrix producing the smallest RMSE would then be used in conjunction with the full sample covariance $S$ to obtain the final estimate $\hat{Q}$.

\subsubsection{Initial guess and convergence}
  
The nonconvex nature of this problem prohibits use of convergence criterion available for common convex optimization problems. Instead, we say that the DC scheme has ``converged'' when $\frac{\|Q_{j+1} - Q_j\|_{F}}{\|Q_j\|_{F}} < \varepsilon$ for some loose tolerance $\varepsilon=0.01$.
In this paper, we used the initial guess $Q_0 = I_\ell$, and this choice can produce large diagonal entries $Q_{ii} \gg 0$ since the diagonal penalty weights of $\Lambda$ are set to zero (see results in Section \ref{GEFS}).


\section{Simulation study}

This section contains a set of simulation studies to assess the ability of 
our proposed algorithm to recover unknown precision structures under the model 
(\ref{lineargaussianmodel}). The section is broken into two classes of basis functions -- localized bases 
whose support is spatially compact, and global basis functions that are 
nonzero over the entire domain. For each class we entertain multiple types of 
precision structures that are common in the graphical modeling literature.

For each choice of $n$, $\Phi$, and $Q$, the noise-to-signal ratio
$\tau^2/(\text{tr}( \Phi Q^{-1} \Phi^{\mathrm{T}})/n)$ 
is fixed at $0.1$ and hence determines the true nugget variance $\tau^2$. Our estimated nugget variance $\hat{\tau}^2$ is retrieved from Section \ref{nug_estimate}.
Although in these simulations the population mean is zero, in practice it is unknown, so throughout we use the standard unbiased estimator $S$ that includes an empirical demeaning which will reflect practical implications better than using the known mean.

\subsection{Local basis} \label{localbasisstudy}

First, we consider a localized problem where we use a basis of compactly 
supported functions on a grid using the LatticeKrig model setup 
\cite{nychka2015}, which we briefly describe here: basis functions are 
compactly supported Wendland functions whose range of support is set so that 
each function overlaps with 2.5 other basis functions along axial directions. 
The model basis functions will correspond 
to either a single level or multiresolution model. 
In the single level setup, functions are placed on a regular grid.  In the 
multiresolution setup, higher levels of resolution are achieved by increasing 
the number of basis functions and nodal points (e.g., the second level 
doubles the number of nodes in each axial direction).  The precision matrix 
$Q$ can be set to approximate Mat\'ern-like behavior, see \cite{nychka2015} 
for details.

We specify the variance of the multiresolution levels to behave like an 
exponential covariance by setting parameter $\nu=0.5$. In LatticeKrig, 
the precision matrix $Q$ is constructed according to a spatial 
autoregression parameterized by the value $\alpha$, which we fix at 
$\alpha = 4.05$. For simplicity, we employ no buffer region when constructing 
the Wendland bases; i.e.~there are no basis functions centered outside of 
the spatial domain. We use the \texttt{R} package \texttt{LatticeKrig} to 
set up the forementioned basis and precision matrices.
A total of $m=500$ realizations from the process 
(\ref{lineargaussianmodel}) under this model are generated, and we repeat this entire spatial data generation process over 30 independent trials.

The spatial domain is $[0,1] \times [0,1]$, and $n$ observation 
locations are chosen uniformly at random in this domain for different sample sizes 
$n \in \{100^2, 150^2, 200^2 \}$. 
For the single level Wendland basis, we use 
$\ell \in \{100,225,400\}$ basis functions. Attempting to mirror these 
dimensions in the multiresolution basis, we use $\ell \in \{119, 234, 404\}$ 
which respectively correspond to 1) four multiresolution levels, the coarsest 
containing 2 Wendland basis functions, 2) three multiresolution levels, the 
coarsest containing 4 Wendland basis functions and 3) four multiresolution levels, 
the coarsest containing 3 Wendland basis functions. 

We parameterize the penalty matrix $\Lambda$ according to 
\begin{equation} \label{constantlambda}
  \Lambda_{ij} = \begin{cases} \lambda, & i \neq j, \\
  0, & i=j, \end{cases}
\end{equation}
allowing for free estimates of the marginal precision parameters. 
A 5-fold cross validation as described in Section \ref{lambda_estimate} 
is used to select 
a penalty matrix $\Lambda$ from eight equally spaced values from 0.005 to 0.1. 
The optimal value is then used with the full sample 
covariance $S$ in (\ref{QUICproblem}) to give a best guess $\hat Q$ for 
the simulated data. 
To validate our proposed estimation approach we report several summary statistics,
each averaged over the 30 trials: 
the Frobenius norm $\|\hat Q - Q \|_{F} / \|Q\|_{F}$, 
the Kullback-Leibler (KL) divergence 
$\text{tr}( \hat Q Q)- \log \det( \hat Q Q)  -  \ell$, 
the percentage of zeros in $Q$ that were missed by $\hat Q$, 
the percentage of nonzero elements in $Q$ that were missed by $\hat Q$, 
the estimated nugget effect $\hat {\tau}^2$, the true nugget effect $\tau^2$, and
the estimated and true negative log-likelihoods $f(\hat Q, \hat {\tau}^2)$ and $f(Q,\tau^2)$, where the function $f$ is defined in Section \ref{nug_estimate}.

\begin{table}[tp]
\caption{Simulation study results for the single level case.  Scores are 
  averaged over 30 independent trials.  Each column represents the number of 
  observation samples, number of basis functions, Frobenius norm, KL divergence, 
  percent of true zeros missed, percent of true nonzeros missed, estimated 
  nugget, true nugget, estimated negative log-likelihood and true negative log-likelihood.
    \label{tab:single.sim.study}}

\bigskip

\centering
\begin{tabular}{c|c||cccccccc}
  \hline
  \hline
  $n$ & $\ell$ & Frob & KL & $\%$MZ & $\%$MNZ & $\hat{\tau}^2$ & $\tau^2$
    & $f(\hat{Q}, \hat{\tau}^2)$ & $f(Q,\tau^2)$ \\ 
  \hline
  & 100 & 0.41 & 6.5 & 21.3\% & 5.3\% & 0.1 & 0.1 & -12678 & -12677 \\
  10000 & 225 & 0.49 & 25.9 & 19.9\% & 2.0\% & 0.1 & 0.1 & -12460 & -12450 \\
  & 400 & 0.72 & 113.4 & 4.7\% & 8.2\% & 0.1 & 0.1 & -12204 & -12220 \\
  \hline
  & 100 & 0.38 & 5.6 & 22.3\% & 5.6\% & 0.1 & 0.1 & -28889 & -28888 \\
  22500 & 225 & 0.43 & 20.8 & 21.1\% & 2.4\% & 0.1 & 0.1 & -28601 & -28591 \\
  & 400 & 0.65 & 82.4 & 4.9\% & 1.0\% & 0.1 & 0.1 & -28259 & -28277 \\
  \hline
  & 100 & 0.37 & 5.3 & 22.6\% & 5.6\% & 0.1 & 0.1 & -51631 & -51630 \\
  40000 & 225 & 0.41 & 18.7 & 21.7\% & 2.6\% & 0.1 & 0.1 & -51287 & -51277 \\
  & 400 & 0.61 & 69.9 & 4.9\% & 1.3\% & 0.1 & 0.1 & -50876 & -50896 \\
  \hline
\end{tabular}
\end{table}

\begin{table}[tp]
\caption{Simulation study results for the multiple level case.  Scores are 
  averaged over 30 independent trials.  Each column represents the number of 
  observation samples, number of basis functions, Frobenius norm, KL divergence, 
  percent of true zeros missed, percent of true nonzeros missed, estimated 
  nugget, true nugget, estimated negative log-likelihood and true negative log-likelihood.
  at the truth.
  \label{tab:mult.sim.study}}

\bigskip

\centering
\begin{tabular}{c|c||cccccccc}
  \hline
  \hline
  $n$ & $\ell$ & Frob & KL & $\%$MZ & $\%$MNZ & $\hat{\tau}^2$ & $\tau^2$
    & $f(\hat{Q}, \hat{\tau}^2)$ & $f(Q,\tau^2)$ \\ 
  \hline
  & 119 & 0.82 & 645 & 5.6\% & 6.8\% & 0.1 & 0.1 & -12866 & -12865 \\
  10000 & 234 & 0.89 & 2024.6 & 6.4\% & 3.7\% & 0.1 & 0.1 & -12733 & -12730 \\
  & 404 & 0.85 & 3789.9 & 0.07\% & 2.5\% & 0.1 & 0.1 & -12719 & -12721 \\
  \hline
  & 119 & 0.77 & 639.7 & 7.1\% & 6.4\% & 0.1 & 0.1 & -29105 & -29105 \\
  22500 & 234 & 0.85 & 2045.4 & 10\% & 3.3\% & 0.1 & 0.1 & -28935 & -28930 \\
  & 404 & 0.93 & 4276.9 & 0.08\% & 2.5\% & 0.1 & 0.1 & -28902 & -28907 \\
  \hline
  & 119 & 0.79 & 634.8 & 7.8\% & 6.0\% & 0.1 & 0.1 & -51866 & -51866 \\
  40000 & 234 & 0.84 & 1943.2 & 11.4\% & 3.0\% & 0.1 & 0.1 & -51663 & -51658 \\
  & 404 & 0.94 & 4438.1 & 0.09\% & 2.5\% & 0.1 & 0.1 & -51612 & -51619 \\
  \hline
\end{tabular}
\end{table}

\subsubsection{Comments} \label{sim_comments}

Tables \ref{tab:single.sim.study} and \ref{tab:mult.sim.study} contain 
results from this simulation study. 
We see that estimating the nugget effect $\tau^2$ by treating the process as 
stationary is quite accurate. Estimates under the multiresolution basis are 
clearly lackluster when compared to the single resolution counterpart. 
The Frobenius norm and KL divergence tend to increase with the size of $\ell$, but this is to be expected as the dimensions of the target precision matrix $Q$ grow in $\ell$. 
The percentage of zeros in $Q$ that are missed (i.e.\ nonzero) in 
$\hat Q$ drops sharply as $\ell$ increases to 400, but this is the consequence of a harsher penalty weight matrix selected in the cross validation scheme.
\subsection{Global basis}

Next, we consider a spatial basis defined globally, that is, without compact 
support. In particular, we set up a harmonic basis via the model
$\phi_i(\bs) = \cos(2\pi  \boldsymbol{\omega}_i^{\mathrm{T}} \mathbf{s})$ 
where the frequencies $\boldsymbol{\omega}_1,\dots,\boldsymbol{\omega}_\ell \in \mathbb{R}^2$ are all pairwise combinations of the form $(\frac{k}{\sqrt{n}}, \frac{j}{\sqrt{n}})$ for $k,j=1,\dots,\sqrt{\ell}-1$.
Given $n$ samples, the corresponding $n \times \ell$ basis matrix $\Phi$ has 
$(i,j)$th entry  $\cos ( 2\pi  \boldsymbol \omega_j^{\mathrm{T}} \bs_i )$. 

Whereas compactly supported basis functions laid on a grid suggest natural 
nearest-neighbor structures for $Q$, in the case of global basis functions 
it is not as clear what natural models might be.  We consider four 
traditional undirected graphical models from the literature: 
\begin{enumerate}
  \item 
    Random graph: Elements of $Q$ are randomly selected to be 1; they are 
    zero otherwise.
  \item 
    Cluster graph: The diagonal of $Q$ consists of approximately 
    $\ell/20$ block matrices, each (square) block randomly populated with 
    the number 1.
  \item 
    Scale-free graph: Generated from the algorithm of \cite{BA}. 
    The nonzero elements are again equal to 1.
  \item 
    Band graph: $Q$ is tridiagonal with its nonzero entries equal to 1.
\end{enumerate}
Figure \ref{fig:graph_structures} contains example precision matrices from 
each of these specifications. 
\begin{figure}[t]
  \centering
  \includegraphics[scale=.5]{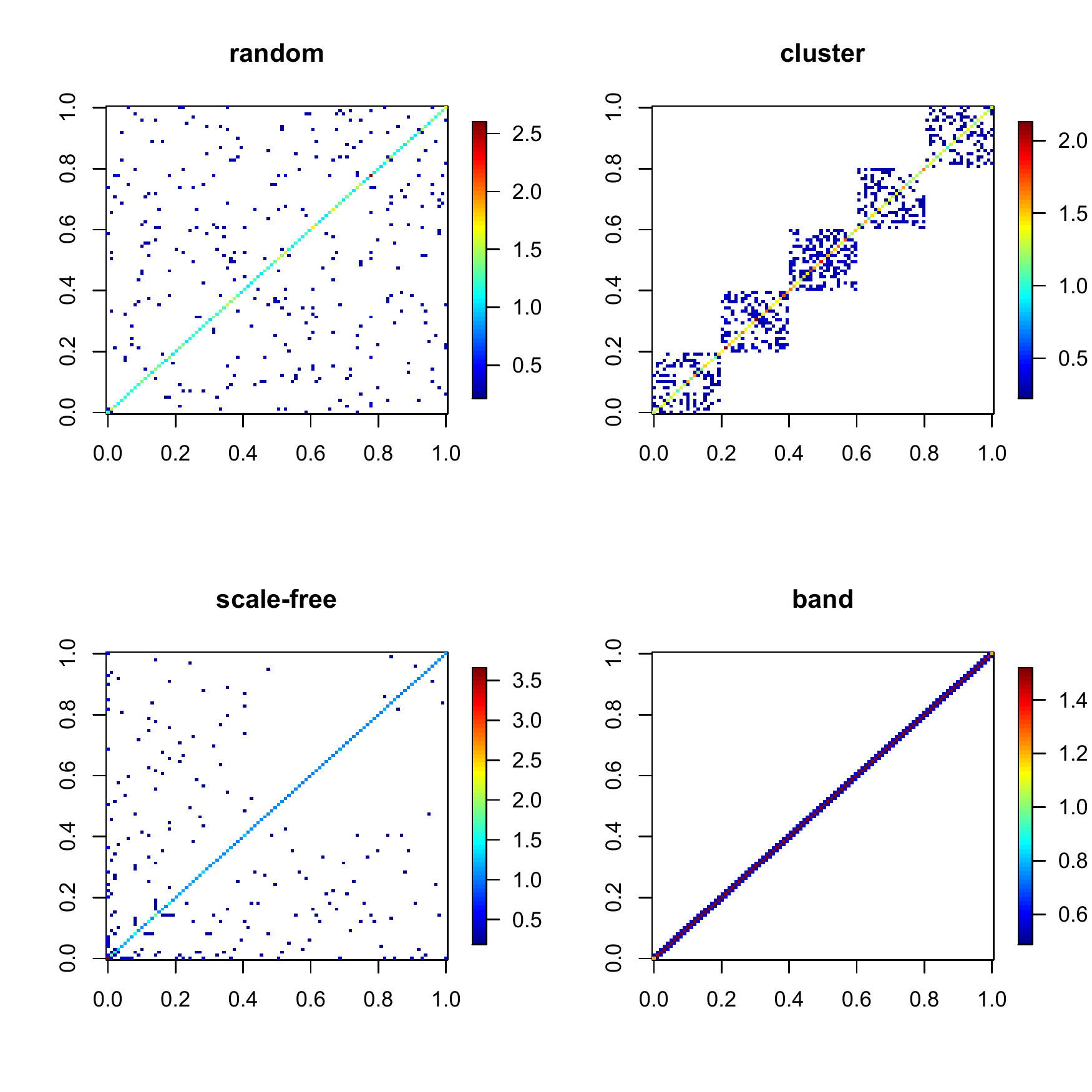}
  \caption{Illustration of the various graph structures of the precision matrix used in our simulation study.}
  \label{fig:graph_structures}
\end{figure}

Our experiment is similar to the local basis experiment: 
the $n$ spatial observation locations are randomly uniformly sampled from the 
square $[0,\sqrt{n}] \times [0, \sqrt{n}]$ where we entertain 
$n \in \{100^2, 150^2,200^2\}$ and $\ell \in \{100,225,400\}$. 
We use the \texttt{R}  package \texttt{huge} \citep{huge} to randomly generate 
the four enumerated precision matrices above. Each graph is generated using default parameters.
The same test statistics as reported in the previous section, again averaged 
over 30 trials, are recorded in Tables \ref{tab:band}-\ref{tab:scalefree}.

\begin{table}[tp]
\caption{Simulation study results for the Band graphical model.  Scores are 
  averaged over 30 independent trials.  Each column represents the number of 
  observation samples, number of basis functions, Frobenius norm, KL divergence, 
  percent of true zeros missed, percent of true nonzeros missed, estimated 
  nugget, true nugget, estimated negative log-likelihood and true negative log-likelihood.
  \label{tab:band}}

\bigskip

\centering

\begin{tabular}{c|c||cccccccc}
  \hline
  \hline
  $n$ & $\ell$ & Frob & KL & $\%$MZ & $\%$MNZ & $\hat{\tau}^2$ & $\tau^2$
    & $f(\hat{Q}, \hat{\tau}^2)$ & $f(Q,\tau^2)$ \\ 
  \hline
        & 100 & 0.17 & 1.5 & 8.5\% & 0\% & 5.0  & 5.0  & 26853 & 26854 \\
  10000 & 225 & 0.19 & 4.2 & 4.6\% & 0\% & 11.3 & 11.3 & 35553 & 35555 \\
        & 400 & 0.21 & 8.7 & 2.4\% & 0\% & 20.0 & 20.0 & 42079 & 42081 \\
  \hline
          & 100 &  0.17 & 1.5 & 8.5\% & 0\% & 5.1 & 5.1 & 59691 & 59692 \\
    22500 & 225 & 0.19 & 4.2 & 4.7\% & 0\% & 11.3 & 11.3 & 78562 & 78564 \\
          & 400 & 0.20 & 8.6 & 2.4\% & 0\% & 20.1 & 20.1 & 92401 & 92402 \\
  \hline
        & 100 & 0.17 & 1.5 & 8.6\% & 0\%  & 5.0 & 5.0 & 105563 & 105564 \\
  40000 & 225 & 0.19 & 4.2 & 4.7\% & 0\%  & 11.3 & 11.3 & 138598 & 138600 \\
        & 400 & 0.20 & 8.6 & 2.4\%  & 0\%  & 20.0 & 20.0 & 162575 & 162576 \\
  \hline
\end{tabular}
\end{table}

\begin{table}[tp]
\caption{Simulation study results for the Cluster graphical model.  Scores are 
  averaged over 30 independent trials.  Each column represents the number of 
  observation samples, number of basis functions, Frobenius norm, KL divergence, 
  percent of true zeros missed, percent of true nonzeros missed, estimated 
  nugget, true nugget, estimated negative log-likelihood and true negative log-likelihood.
  \label{tab:cluster}}

\bigskip

\centering
\begin{tabular}{c|c||cccccccc}
  \hline
  \hline
  $n$ & $\ell$ & Frob & KL & $\%$MZ & $\%$MNZ & $\hat{\tau}^2$ & $\tau^2$
    & $f(\hat{Q}, \hat{\tau}^2)$ & $f(Q,\tau^2)$ \\ 
  \hline
        & 100 & 0.26 & 3.1 & 16.3\% & 0\% & 5.1 & 5.1 & 26867 & 26868  \\
  10000 & 225 & 0.30 & 8.9 & 8.9\% & 0\% & 11.3 & 11.3 & 35573 & 35576 \\
        & 400 &  0.32 & 18 & 5\% & 0\% & 20.0 & 20.0 & 42115 & 42118  \\
  \hline
          & 100 & 0.26 & 3.1 & 17.1\% & 0\% & 5.1 & 5.1 & 59690 & 59691  \\
    22500 & 225 & 0.29 & 8.6 & 9.0\% & 0\% & 11.3 & 11.3 & 78567 & 78570  \\
          & 400 & 0.31 & 17.8 & 5.0\% & 0\% & 20.0 & 20.0 & 92420 & 92423  \\
  \hline
        & 100 &  0.26 & 3.1 & 17.5\% & 0\% & 5.1 & 5.1 & 105568 & 105570  \\
  40000 & 225 & 0.29 & 8.7 & 9.1\% & 0\%  & 11.3 & 11.3 & 138625 & 138628  \\
        & 400 & 0.31 & 17.6 & 5.1\%  & 0\%  & 20.0 & 20.0 & 162620 & 162624  \\
  \hline
\end{tabular}
\end{table}

\begin{table}[tp]
\caption{Simulation study results for the Random graphical model.  Scores are 
  averaged over 30 independent trials.  Each column represents the number of 
  observation samples, number of basis functions, Frobenius norm, KL divergence, 
  percent of true zeros missed, percent of true nonzeros missed, estimated 
  nugget, true nugget, estimated negative log-likelihood and true negative log-likelihood.
  \label{tab:random}}

\bigskip

\centering
\begin{tabular}{c|c||cccccccc}
  \hline
  \hline
  $n$ & $\ell$ & Frob & KL & $\%$MZ & $\%$MNZ & $\hat{\tau}^2$ & $\tau^2$
    & $f(\hat{Q}, \hat{\tau}^2)$ & $f(Q,\tau^2)$ \\ 
  \hline
        & 100 & 0.19 & 1.7 & 9.6\% & 0\% & 5.1 & 5.1 & 26872 & 26873  \\
  10000 & 225 & 0.20 & 4.8 & 5.2\% & 0\% & 11.3 & 11.3 & 35586 & 35589 \\
        & 400 &  0.22 & 9.8 & 2.6\% & 0\% & 20.0 & 20.0 & 42134 & 42136 \\
  \hline
        & 100 & 0.19 & 1.8 & 9.7\% & 0\% & 5.1 & 5.1 & 59697 & 59698  \\
    22500 & 225 & 0.20 & 4.8 & 5.2\% & 0\% & 11.3 & 11.3 & 78581 & 78584  \\
        & 400 & 0.22 & 9.7 & 2.6\% & 0\% & 20.0 & 20.0 & 92441 & 92443  \\
  \hline
        & 100 &  0.19 & 1.7 & 9.8\% & 0\%  & 5.1 & 5.1 & 105585 & 105586  \\
  40000 & 225 & 0.20 & 4.7 & 5.2\%  & 0\%  & 11.3 & 11.3 & 138643 & 138645  \\
        & 400 & 0.22 & 9.7 & 2.6\%  & 0\%  & 20.0 & 20.0 & 162640 & 162642 \\
  \hline
\end{tabular}
\end{table}

\begin{table}[tp]
\caption{Simulation study results for the Scale-free graphical model.  Scores are 
  averaged over 30 independent trials.  Each column represents the number of 
  observation samples, number of basis functions, Frobenius norm, KL divergence, 
  percent of true zeros missed, percent of true nonzeros missed, estimated 
  nugget, true nugget, estimated negative log-likelihood and true negative log-likelihood.
  \label{tab:scalefree}}

\bigskip

\centering
\begin{tabular}{c|c||cccccccc}
  \hline
  \hline
  $n$ & $\ell$ & Frob & KL & $\%$MZ & $\%$MNZ & $\hat{\tau}^2$ & $\tau^2$
    & $f(\hat{Q}, \hat{\tau}^2)$ & $f(Q,\tau^2)$ \\ 
  \hline
        & 100 & 0.21 & 1.4 & 6.1\% & 0\% & 5.1 & 5.1 & 26879 & 26879 \\ 
  10000 & 225 & 0.21 & 3.5 & 2.8\% & 0\% & 11.3 & 11.3 & 35604 & 35605 \\ 
        & 400 &  0.21 & 6.3 &  2.6\% & 0\% & 20.0 & 20.0 & 42167 & 42172 \\ 
  \hline
        & 100 & 0.20 & 1.4 & 5.9\% & 0\% & 5.1 & 5.1 & 59698 & 59699  \\ 
  22500 & 225 & 0.21 & 3.5 & 2.8\% & 0\% & 11.3 & 11.3 & 78598 & 78599 \\ 
        & 400 & 0.21 & 6.3 & 2.7\% & 0\% & 20.0 & 20.0 & 92471 & 92476 \\ 
  \hline
        & 100 &   0.20 & 1.3 & 6.1\% & 0\%  & 5.1 & 5.1 & 105597 & 105598 \\ 
  40000 & 225 & 0.21 & 3.5 & 2.8\%  & 0\%  & 11.3 & 11.3 & 138661 & 138662 \\ 
        & 400 & 0.21 & 6.2 & 2.7\%  & 0\%  & 20.0 & 20.0 & 162671 & 162676 \\ 
  \hline
\end{tabular}
\end{table}

\subsubsection{Comments}

Tables \ref{tab:band}, \ref{tab:cluster}, \ref{tab:random}, and 
\ref{tab:scalefree} contain the results of the global basis simulation study. 
As in the compactly supported basis study we note the same behavior in 
$\hat{\tau}^2$ in that the independent identical coefficient assumption (constant diagonal $Q$) 
yields robust estimates of $\tau^2$. 
The cluster graph appears highest in Frobenius norm and KL divergence, but this 
should not be surprising given the fact that there are more nonzero elements in 
the cluster graph than its counterparts and we are searching for a sparse 
estimate. The percentage of missed zeros and missed nonzeros in $Q$ also 
behave similarly to the local basis study with respect to the dimension $\ell$. 
Overall the proposed method seems to supply reasonable estimates of the zero 
structure of the precision matrix, as well as its non-zero values.

\subsection{Changing the Ensemble Size and the Noise-to-Signal Ratio}

We consider a brief study to assess the effect of ensemble size, or 
field realizations, on the algorithm's ability to recover $Q$. In the prior 
section, we fixed the ensemble size at $m=500$. Here we fix $\ell=100$ to ease 
computation times but vary the ensemble size according to 
$m \in \{50,100,200,500,1000\}$ and number of spatial locations according to 
$n \in \{100^2, 150^2, 200^2, 250^2, 300^2 \}$. The same Wendland basis and 
precision matrix $Q$ as in the local basis study (\ref{localbasisstudy}) are used 
to generate $m$ realizations of the additive model (\ref{lineargaussianmodel}) 
where we have observations at $n$ uniformly randomly sampled locations in 
$[0,1] \times [0,1]$. We record the Frobenius norm $\|\hat Q - Q \|_F / \|Q\|_F$, the KL divergence $\text{tr}( \hat Q Q)- \log \det( \hat Q Q)  -  \ell$,
and the percentage of zeros in $Q$ that $\hat{Q}$ fails to capture. The penalty parameter is fixed at $\lambda = 0.005$, the value which was favored in our previous simulations when $(\ell,m,n) = (100,500,10000)$.

\begin{figure}[t]
  \centering
  \includegraphics[scale=.41]{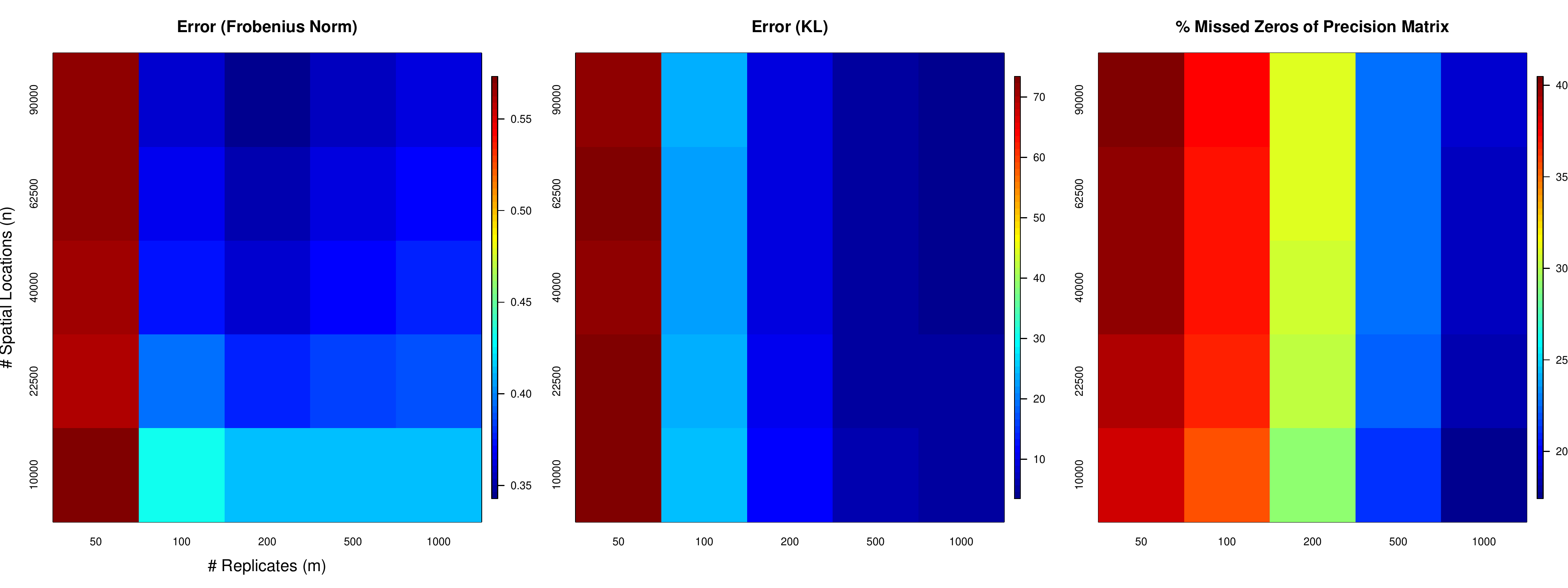}
  \caption{Results of a small simulation study where we fix the basis size $\ell$ but 
  vary sample size $m$ and number of spatial locations $n$. The penalty matrix is fixed across each pair $(n,m)$.}
  \label{fig:sim.ensemble}
\end{figure}

Figure \ref{fig:sim.ensemble} shows results from this study. The plots suggest that the number of replicates $m$ has a prominent effect on our summary statistics, 
while the influence of the spatial dimension $n$ is more subtle.  The Frobenius norm and KL divergence decrease as $n$ increases but much more noticeably decrease when
$m$ increases, with the effects less pronounced beyond the $m=50$ range. 
The percentage of missed zeros in $\hat{Q}$ behaves similarly, but the accuracy does not rise so sharply beyond $m=50$. 

We also conducted a small experiment where we changed the noise-to-signal ratio, which has been kept at $0.1$ up to this point, under a fixed set of basis functions and precision matrix $Q$. The penalty matrix was also kept constant across the various noise levels. Increasing the ratio from $0.1$ to $0.25$ slightly worsened the Frobenius norm and KL divergence, but the jump from $0.25$ to $0.5$ demonstrated a substantial decrease in ability to accurately capture the true precision matrix. At a noise-to-signal ratio of 0.5, however, the resulting model is extremely noisy and expecting accurate estimates is not entirely reasonable.

Finally, we note that the estimation procedure for $\hat{\tau}^2$ in Section \ref{nug_estimate} remained extremely accurate throughout all the modifications in this section.

\section{Data Analysis} \label{data}

\subsection{Reforecast Data} \label{GEFS}

The Global Ensemble 
Forecast System (GEFS) from the National Center for Environmental Prediction (NCEP) provides an 11 member daily reforecast of climate variables 
available from December 1984 to present day. We take all readings from 
January of each year through 2018, giving a total of $m=1054$ global fields 
of the two meter temperature variable. Measurements were recorded at each integer 
valued longitude and latitude combination, totaling $n = 65160$ spatial locations. 
We use a set of Wendland basis functions spread over the globe 
in a way that ensures that the $\ell = 2531$ centers are equispaced with respect 
to the great circle distance. See Figure \ref{fig:gefs_day1} for an illustration 
of the data and basis functions.  We opt for 2531 as it provides a reasonably dense network of 
basis functions that still allows for computational tractability. 

\begin{figure}[t]
  \centering
  \includegraphics[scale=.4]{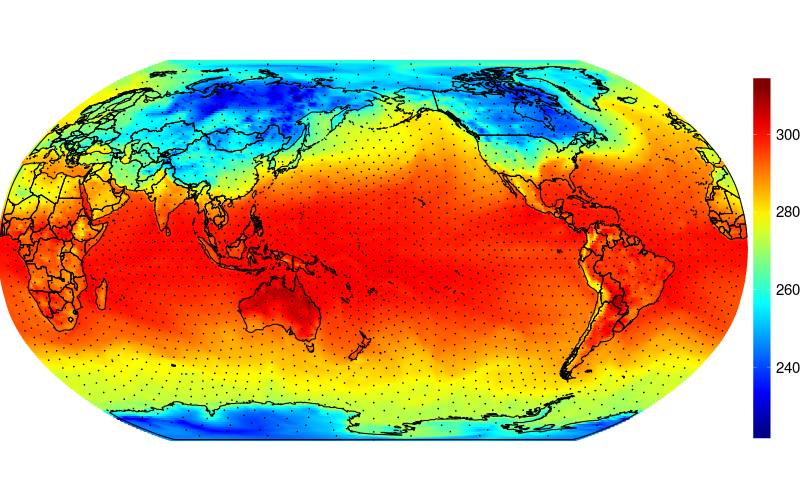}
  \caption{Two meter temperature on January 1, 1984. Dots indicate nodal points.}
    \label{fig:gefs_day1}
\end{figure}
%

Throughout this section, we work with temperature anomalies, that is, residuals 
after removing a spatially-constant mean of approximately 276$^\circ$ K from the data matrix. 
The technique introduced in Section \ref{nug_estimate} is used to estimate the 
nugget effect $\hat{\tau}^2 = 1.74$. 

An interesting idea when using localized basis functions with a notion of 
distance between them is to adjust the penalty matrix so that neighbors are 
encouraged to remain in close proximity to the center point. In particular, 
we make $\Lambda$ proportional to the distance matrix of the centers of the 
basis functions. This idea was pursued in \citep{davanloo} but in the context 
of direct spatial observations with the graphical lasso rather than working 
through basis functions with our DC algorithm. 

We consider $\Lambda = \lambda D$ where $D$ is the pairwise distance matrix 
of the nodal points registering the Wendland basis functions. 
To select the penalty parameter $\lambda$, we use two-fold likelihood-based 
cross-validation as described in Section \ref{lambda_estimate} for 
values 
$\lambda \in \{0.0001,0.0005,\linebreak 0.001,0.005,0.01,0.05,0.1,0.5\}$. 
Smaller values than $0.0001$ were examined but 
failed to converge after a reasonable runtime. Despite the fact that it sits 
on the boundary of our parameter set, the value $\lambda = 0.0001$ was chosen 
as at yielded a significant drop in negative log-likelihood when compared to 
larger penalties. We apply the DC algorithm (\ref{QUICproblem}) a final time with the full sample covariance matrix and the $\Lambda$ favored by cross-validation.
\begin{figure}[t]
  \centering
  \includegraphics[scale=.4]{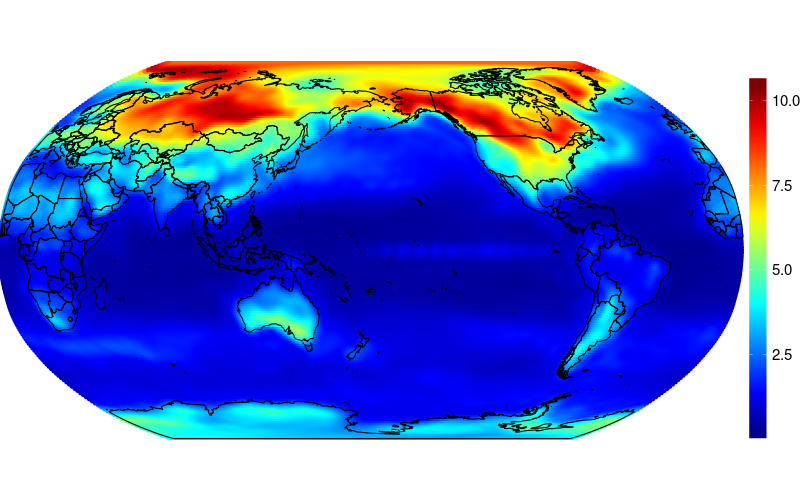}
  \caption{Estimated pointwise standard deviations under the model (\ref{basis}) 
  using $Q$ estimated from the DC algorithm. Units are degrees Celsius.}
  \label{fig:gefs_standarderrors}
\end{figure}

Figure \ref{fig:gefs_standarderrors} contains a global plot of the implied 
estimated local standard deviations. Note similar behavior in Figure 4 
of \cite{Legates1990} which depicts standard deviations for mean air 
surface temperature over the globe. Both plots illustrate a clear 
land-ocean difference and increased variability in higher latitudes where 
the overall land area is greater.

Figure \ref{fig:gefs_spatialcorrelation} shows a plot of estimated spatial 
covariance functions centered at three different geographical locations 
in the southern tip of South America, the Middle East, and central North America. 
There is clear evidence of nonstationarity in all three cases, as well as 
negative correlation at medium distances.  An interesting feature of the 
estimated covariance structure is the negative correlation between Alaska and 
the central United States, indicative of medium-range teleconnections, and 
may be a result of Rossby waves that occur during winter in the northern hemisphere.


\begin{figure}[t] \label{fig:gefs_spatialcorrelation}
  \minipage{0.33\textwidth}
  \includegraphics[width=\linewidth]{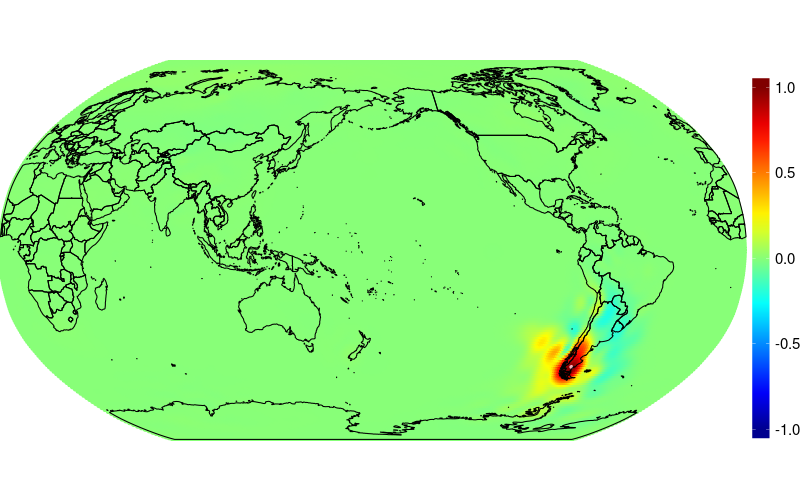}
  \endminipage\hfill
  \minipage{0.33\textwidth}
  \includegraphics[width=\linewidth]{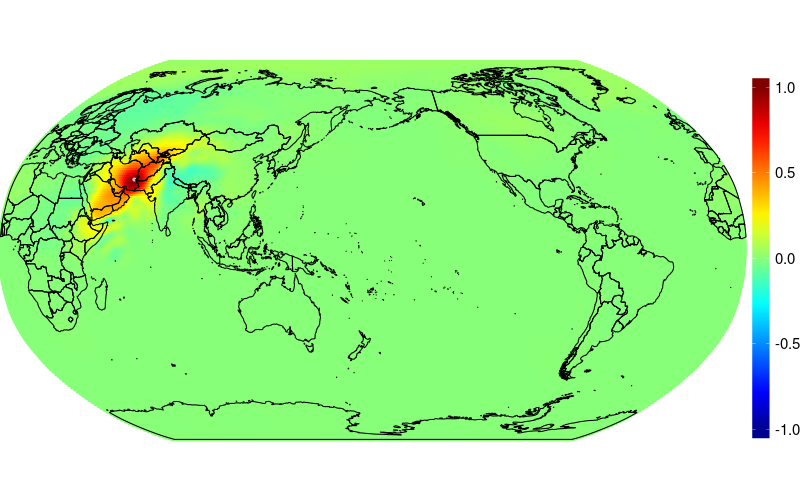}
  \endminipage\hfill
  \minipage{0.33\textwidth}%
  \includegraphics[width=\linewidth]{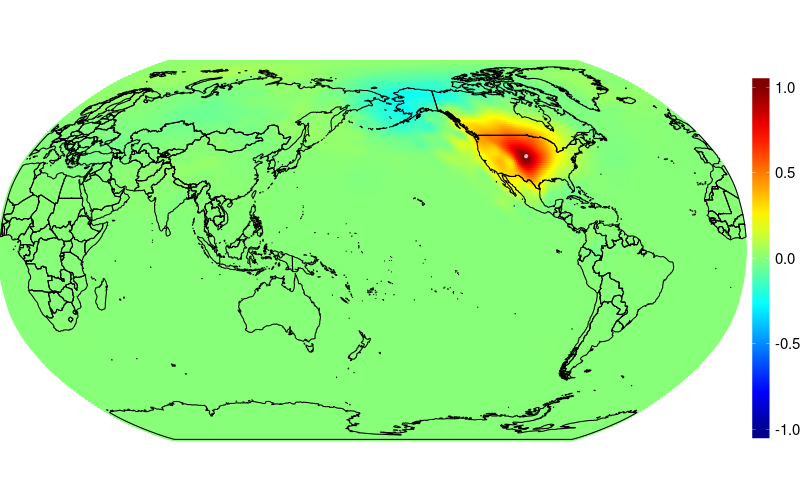}
  \endminipage
  \caption{Estimated spatial correlation functions centered at a point in southern 
  South America (left), central Middle East (middle) and central North America 
  (right).}
  \label{fig:gefs_spatialcorrelation}
\end{figure}

An interesting byproduct of our model is that we can examine the GMRF 
neighborhood structure of the estimated precision matrix. Recall that each 
random coefficient $c_i$ in our model is registered to a nodal grid shown 
in Figure \ref{fig:gefs_day1}, and thus we can identify estimated spatial 
neighborhood patterns according to this grid. 
We illustrate some of these neighborhood structures, colored by their 
respective $Q$ values in Figure \ref{fig:gefs_neighbors}. 
For nodal points over the ocean, the tendency is large, positive precision 
values with few neighbors. For nodal points over land, we observe that the 
most significantly nonzero neighbor elements are geographically near the 
center node, partly due to our choice for $\Lambda$.  There is sometimes 
evidence of neighborhoods that spread throughout the globe, although the 
magnitude of the corresponding precision matrix entries is typically very small, 
e.g., the lower right panel of Figure \ref{fig:gefs_neighbors}.
\begin{figure}[t] 
  \label{fig7} 
  \begin{minipage}[b]{0.5\linewidth}
    \centering
    \includegraphics[width=.95\linewidth]{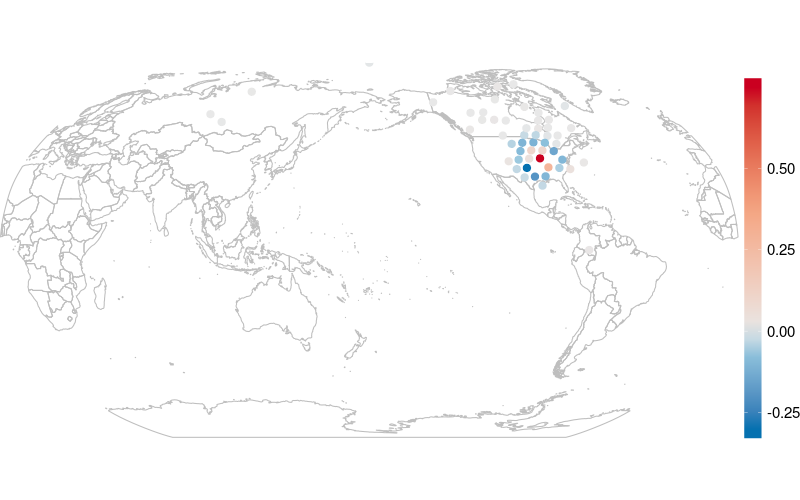} 
  \end{minipage}
  \begin{minipage}[b]{0.5\linewidth}
    \centering
    \includegraphics[width=.95\linewidth]{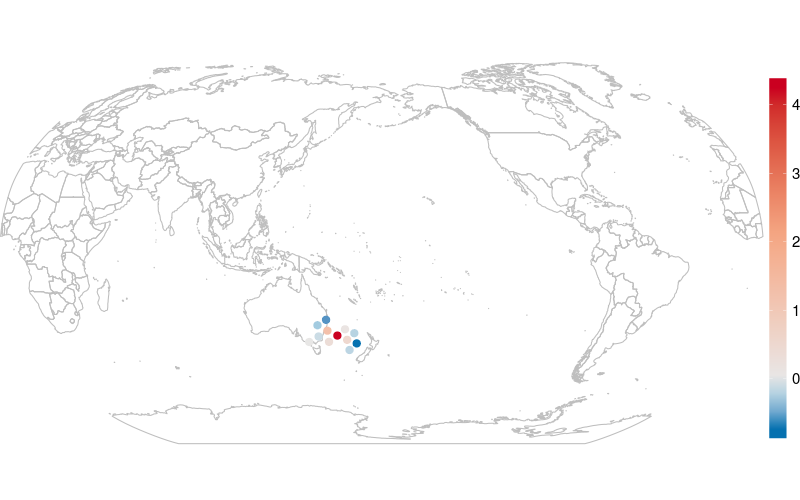} 
  \end{minipage} 
  \begin{minipage}[b]{0.5\linewidth}
    \centering
    \includegraphics[width=.95\linewidth]{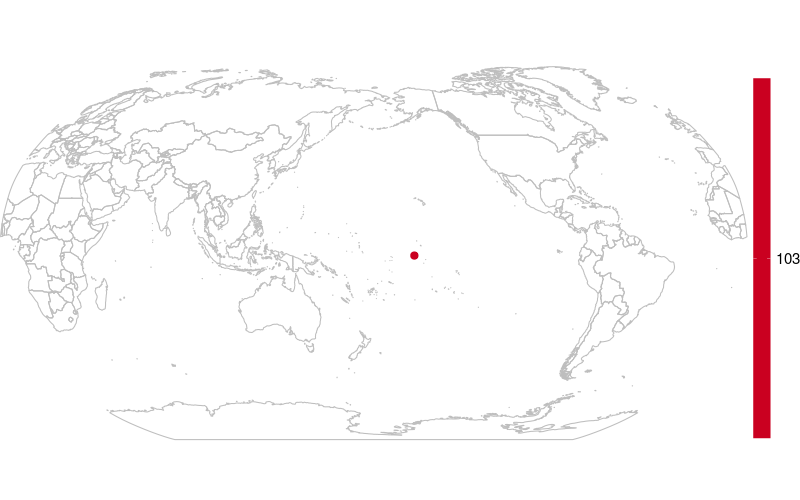} 
  \end{minipage}
  \begin{minipage}[b]{0.5\linewidth}
    \centering
    \includegraphics[width=.95\linewidth]{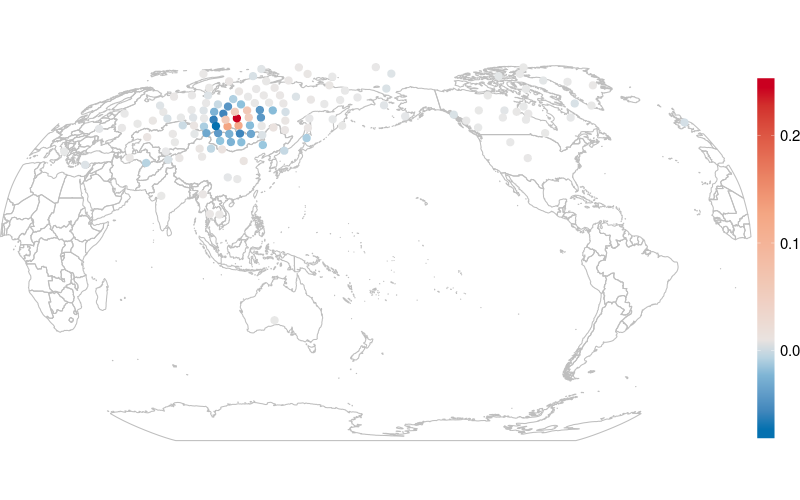} 
  \end{minipage} 
  \caption{The estimated neighborhood structure of $Q$ registered at 
  the nodal points of the basis functions.  Colors correspond to neighborhood 
  entries in $Q$ with a center point that is clockwise starting top left: 
  over U.S.; in Pacific Ocean near Australia; over Russia and over the Pacific 
  Ocean near the equator.}
  \label{fig:gefs_neighbors}
\end{figure}

\subsection{TopoWx}

The Topography Weather (TopoWx) dataset \citep{oyler2015} contains 
observed 2 m temperatures from a set of observation networks over the continental 
United States.  
We consider daily minimum temperatures during the month of June from 2010 to 2014, 
giving a total of $m=150$ replicates. Network locations are chosen to have no 
missing values, yielding $n = 4577$ spatial locations. Figure \ref{fig:topo_basis} 
shows an example day of data on June 1, 2010.
\begin{figure}[t]
  \centering
  \includegraphics[scale=.4]{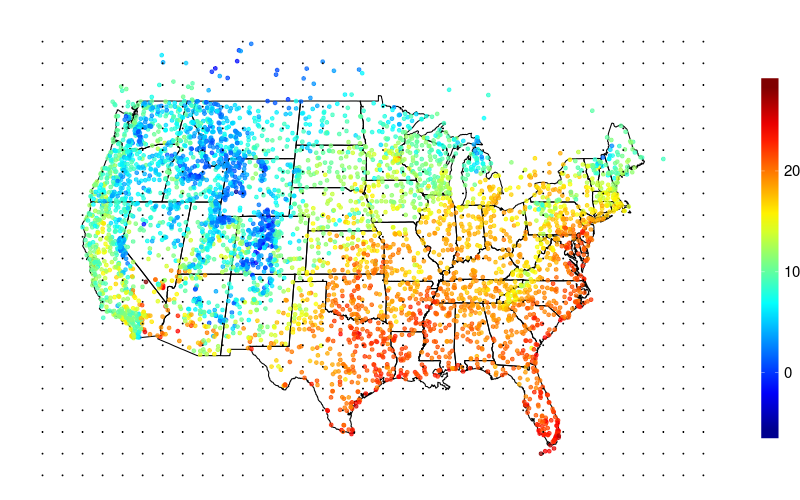}
  \caption{Minimum temperature on June 1, 2010, overlaid with a grid of 
  basis functions nodes.}
  \label{fig:topo_basis}
\end{figure}

The dataset includes an elevation covariate, and we work with minimum temperature 
residuals after regressing out a mean function linear in longitude, latitude 
and elevation. We transform the raw spatial coordinates with a sinusoidal 
projection. 
Our statistical model for the temperature residuals uses Wendland basis 
functions centered at nodes displayed in Figure \ref{fig:topo_basis}. We opt for 
$\ell=1160$ functions using a single level of resolution. The nodal grid and 
Wendland functions are chosen to match up with a LatticeKrig model specification, 
but with relaxed assumptions on the precision matrix governing the random 
coefficients. 

The nugget estimate $\hat{\tau}^2 = 2.18$  is retrieved as in Section \ref{nug_estimate}. 
As with the previous dataset, the penalty matrix $\Lambda$ is 
parameterized according to $\Lambda = \lambda D$, where $D$ is the distance 
matrix of node points. We selected scaling parameters 
from $\lambda \in \{ \frac{1}{10},\frac{2}{10},\dots,\frac{9}{10},1,
\frac{10}{9},\dots,\frac{10}{2},10 \} \cup \{ 10,12.5,\dots,27.5, 30 \}$ 
using a simple prediction-based cross validation scheme in Section 
\ref{lambda_estimate}; the latter set in the union was considered to check 
the behavior of the final estimate $\lambda = 10$ near the boundary. 
The resulting penalty matrix is used with $\hat{\tau}^2$ and the full 
sample covariance $S$ to obtain a final estimate of the precision matrix. 

Figures \ref{fig:topo_1level_neighborhood} and \ref{fig:topo_1level_correlation}  
show graphical model neighborhoods and estimated correlation functions  
centered at locations in Utah and Kansas.  Clear anisotropy 
and nonstationarity is present in the estimated correlation functions with 
greater north-south directionality of correlation, while the neighborhood 
structure for the Utah nodal point displays greater complexity than the relatively 
nearby neighbors of the nodal point in the midwest.

\begin{figure}[t] 
  \label{ fig7} 
  \begin{minipage}[b]{0.5\linewidth}
    \centering
    \includegraphics[width=.95\linewidth]{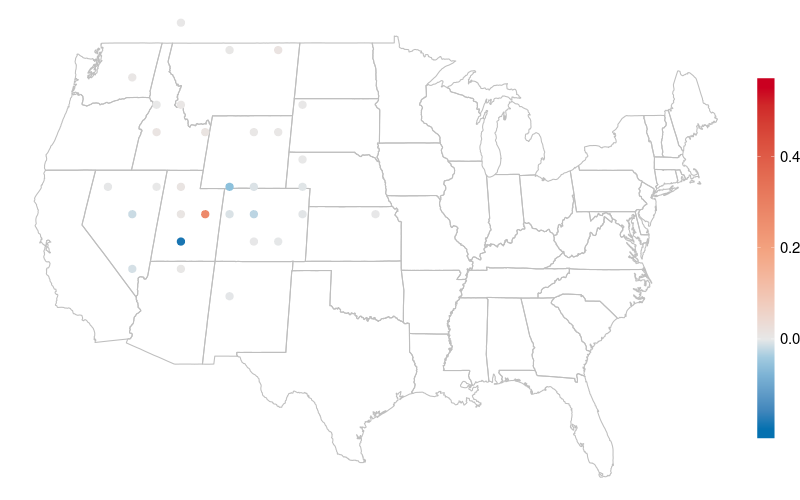} 
  \end{minipage}
  \begin{minipage}[b]{0.5\linewidth}
    \centering
    \includegraphics[width=.95\linewidth]{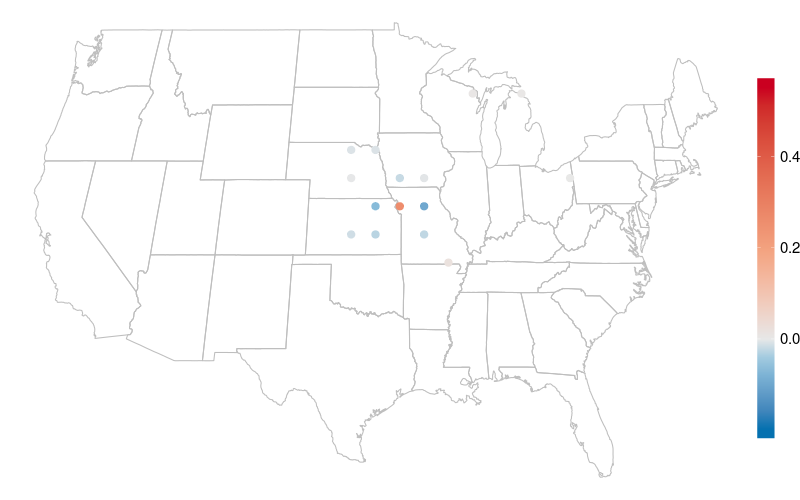} 
  \end{minipage} 
  \caption{Displaying the neighborhood structure of $Q$, where the two center 
  points are located in central Utah and near Kansas City. Neighbors are 
  colored according to the corresponding nonzero elements in $Q$.}
  \label{fig:topo_1level_neighborhood}
\end{figure}
\begin{figure}[t] 
  \label{ fig7} 
  \begin{minipage}[b]{0.5\linewidth}
    \centering
    \includegraphics[width=.95\linewidth]{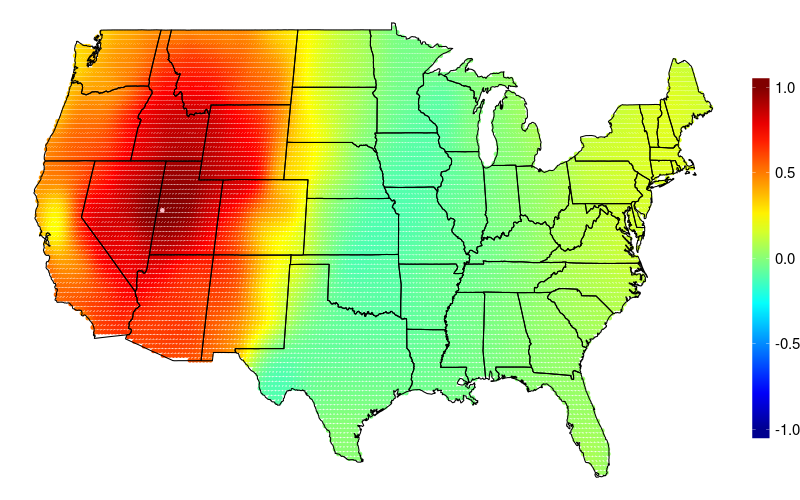} 
  \end{minipage}
  \begin{minipage}[b]{0.5\linewidth}
    \centering
    \includegraphics[width=.95\linewidth]{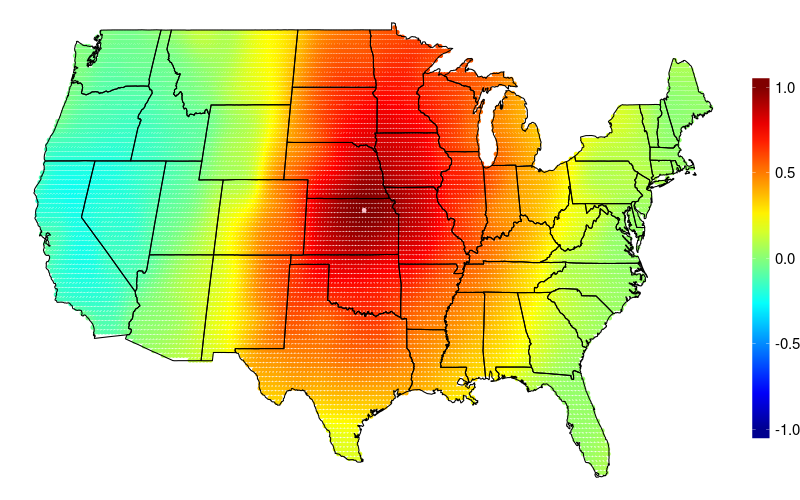} 
  \end{minipage} 
  \caption{Estimated spatial correlation functions centered at the pink-colored locations 
  in Utah and Kansas.}
  \label{fig:topo_1level_correlation}
\end{figure}

Due to lack of data availability over the ocean there is an identifiability 
problem with our method. Since several of the Wendland basis functions lie over 
the ocean where there is no observed data, we cannot expect the algorithm to give
reasonable estimates for the diagonal elements of $Q$ corresponding 
to those nodes. Moreover, the diagonal of the penalty  matrix $\Lambda$ is 
identically zero, and thus the corresponding diagonal elements of $Q$ remained 
unchanged no matter the initial guess $Q_0$, which we fixed at $Q_0 = I_\ell$.

Now we entertain a comparison against the corresponding LatticeKrig model using 
the same nodal grid and same Wendland basis functions, but with the spatial 
autoregressive precision matrix of LatticeKrig. We estimate LatticeKrig 
parameters by maximum likelihood within the \texttt{LatticeKrig} package in 
\texttt{R}. In particular, the estimated nugget variance $\hat{\tau}^2 = 4.7$ is about twice as large as that from our model, and the central \texttt{a.wght} parameter is estimated at 5.41. The smoothness parameter $\nu$ in the LatticeKrig setup is set at $0.5$ which is a typical assumption for observational temperature data.

We compare the two models based on cross-validation prediction accuracy 
and standard Akaike information criterion. 
For 400 randomly held-out locations, we calculate predictive 
squared error (MSE) and continuous ranked probability score (CRPS) \citep{CRPS} 
based on the standard kriging predictor separately for each of 150 days of data. 
A nice consequence of the spatial 
basis model (\ref{basis}) is that the formulas for kriging predictors and 
kriging variances can be evaluated in $\mathcal{O}(n \ell^2)$ operations rather 
than the naive $\mathcal{O}(n^3)$; see e.g.\ \citep{cressie2008} for details. 
The $400 \times 150 = 60000$ MSE and CRPS values are then averaged into single 
summary statistics. For the proposed model, the MSE and CRPS scores are 4.99 
and 3.73, respectively, while the values are 5.02 and 3.64 for 
LatticeKrig. The similarity of these statistics between models should not 
come as a surprise given the density of spatial observations.

Direct likelihood comparisons between our proposal and LatticeKrig is not 
fair due to the high number of free parameters in our model. For 
spatial processes, the number of degrees of freedom of the model can be 
identified with the trace of the spatial smoothing hat matrix \citep{nychka2000}. 
Using this estimate of effective degrees of freedom as the number of parameters, 
our model has AIC of 9043 while LatticeKrig has AIC of 12615, indicating 
substantial improvement in model fit using our proposal.

\subsubsection{Including Global Basis Functions}

In the previous section, we removed a mean trend from our data by regressing 
out longitude, latitude and elevation. Notice, however, that we could simply 
choose to include for example, elevation, as a global basis function in the 
stochastic part of the spatial model (\ref{basis}). In light of this viewpoint, 
we regress out a mean trend on only longitude and latitude and use the same 
single resolution Wendland basis as before (with $\ell = 1160$) but also along 
with elevation as a global basis function. The penalty matrix $\Lambda$ is 
selected from a prediction-based cross validation over parameters 
$(\lambda,\gamma)$, where the upper-left left block of $\Lambda$ reflects 
distance between the Wendland bases via $\lambda D$, but now the 
$1161^{st}$ row/column of $\Lambda$ is held at a constant $\gamma$ since 
there is no natural notion of geographical 
distance between our global and locally registered basis functions. 

By way of illustration, Figure \ref{fig:newbasis_correlation} shows 
estimated spatial correlations using elevation in the stochastic part of the model, 
which has a noticeable effect. In particular, there are clear estimated spatial 
correlation patterns with the western point that interact with the nontrivial 
topography of the western US. The spatial correlation structure in the eastern 
US is much more homogeneous, reflecting the comparative lack of complex terrain.

An interesting byproduct of including elevation as a globally defined predictor, 
while the rest of the stochastic basis functions are compactly supported and 
registered to a grid, is that we can examine the nodal neighbors of the 
elevation coefficient. To be precise, this corresponds to the first 1160 elements of the $1161^{st}$ row/column of $Q$, as shown in 
Figure \ref{fig:newbasis_elevation_neighbor}. 
The neighbors, especially those that are substantially different from zero, 
are concentrated in the western US and over the Appalachian Mountains in the 
east. 
\begin{figure}[t] 
  \label{fig7} 
  \begin{minipage}[b]{0.5\linewidth}
    \centering
    \includegraphics[width=.95\linewidth]{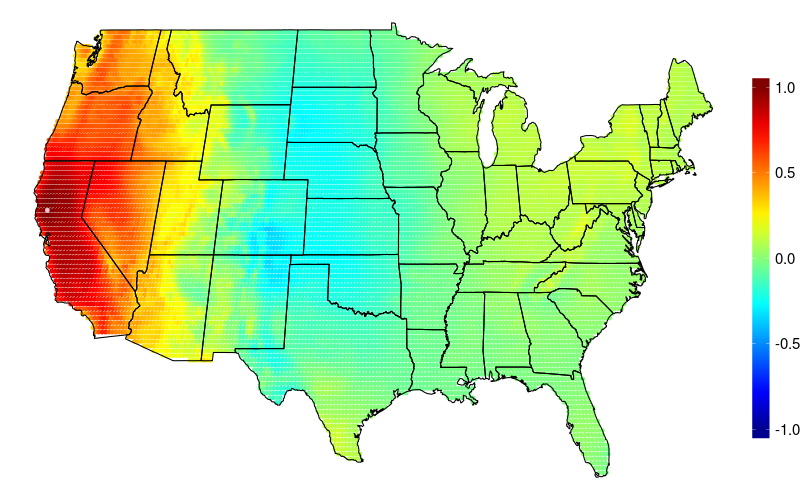} 
  \end{minipage}
  \begin{minipage}[b]{0.5\linewidth}
    \centering
    \includegraphics[width=.95\linewidth]{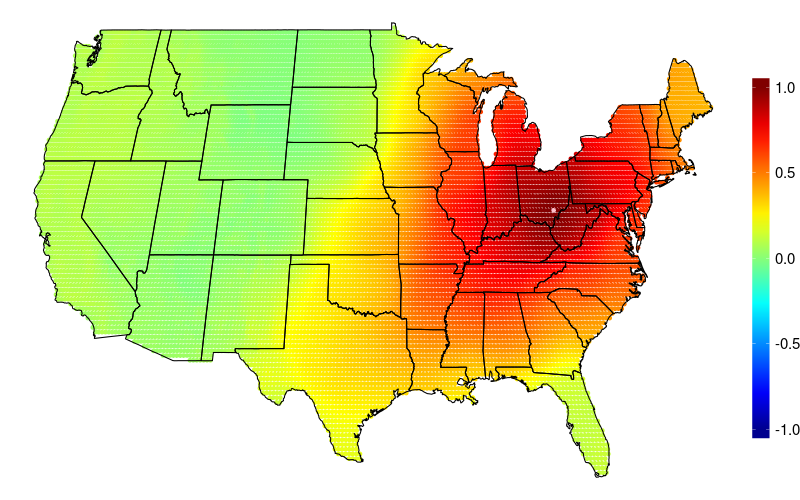} 
  \end{minipage} 
  \caption{Estimated spatial correlation functions of the process registered at a 
  spatial point colored in pink using elevation as basis function in the stochastic 
  model.}
  \label{fig:newbasis_correlation}
\end{figure}
\begin{figure}[t]
  \centering
  \includegraphics[scale=.4]{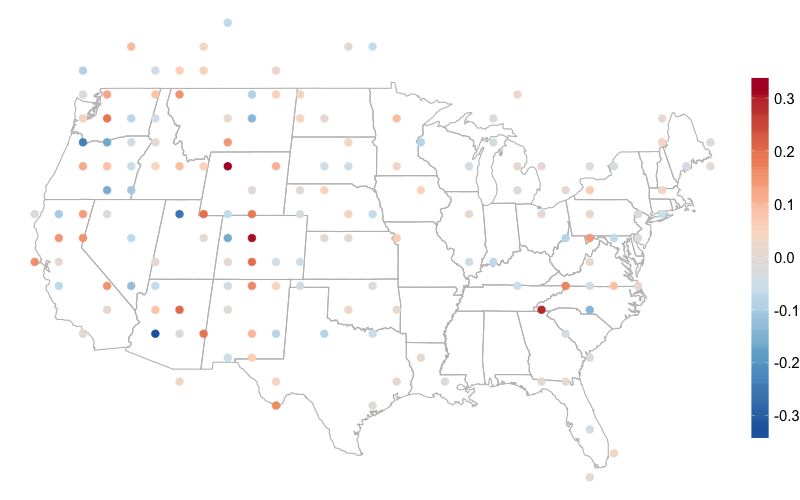}
  \caption{The Wendland coefficient graphical neighbors of the elevation coefficient.}
  \label{fig:newbasis_elevation_neighbor}
\end{figure}

\section{Conclusion}

In this work we introduce a novel approach for estimating the precision matrix 
of the random coefficients of a basis representation model that is pervasive 
in the spatial statistical literature. The only assumption we enforce is that 
the precision matrix is sparse.  In the case that the basis functions are registered 
to a grid, the precision entries can be interpreted as a spatial Gaussian 
Markov random field, while graphical model interpretations are still viable 
with global bases. 

The estimator minimizes a penalized log likelihood, and we show that the 
optimization problem is equivalent to one involving a sum of a convex and 
concave functions, which suggests a DC-algorithm in which 
we iteratively linearize the concave part at the previous guess and 
solve the resulting convex problem. 
The linearization in our case gives rise to a graphical lasso problem with 
its ``sample covariance'' depending upon the previous guess. The graphical 
lasso problem is well studied and a number of user-friendly \texttt{R} 
packages exist, headed by the second order method \texttt{QUIC}. 
Our method has important practical applications in spatial data analysis, 
since we obtain a nonparametric, penalized maximum likelihood estimate 
of $Q$ which can subsequently be used in kriging or simulation 
with computational complexity $\mathcal{O}(n\ell^2)$ under the basis model. 

In our data examples we see that the proposed method performs competitively 
with existing alternatives such as LatticeKrig, but substantially improves 
information criteria. Moreover, our model results in highly interpretable 
fields, allowing for checking of graphical neighborhood structures, or 
implied nonstationary covariance functions.  Future work may be directed 
toward other penalties, or pushing these notions though to space-time modeling.

\bibliographystyle{plainnat}

\begin{thebibliography}{33}
\providecommand{\natexlab}[1]{#1}
\providecommand{\url}[1]{\texttt{#1}}
\expandafter\ifx\csname urlstyle\endcsname\relax
  \providecommand{\doi}[1]{doi: #1}\else
  \providecommand{\doi}{doi: \begingroup \urlstyle{rm}\Url}\fi

\bibitem[Bandyopadhyay and Lahiri(2009)]{bandy2010}
Soutir Bandyopadhyay and Soumendra~N. Lahiri.
\newblock Asymptotic properties of discrete {F}ourier transforms for spatial
  data.
\newblock \emph{Sankhy\={a}}, 71\penalty0 (2, Ser. A):\penalty0 221--259, 2009.
\newblock ISSN 0972-7671.

\bibitem[Banerjee et~al.(2008{\natexlab{a}})Banerjee, El~Ghaoui, and
  d'Aspremont]{banerjee_model_2008}
Onureena Banerjee, Laurent El~Ghaoui, and Alexandre d'Aspremont.
\newblock Model selection through sparse maximum likelihood estimation for
  multivariate {G}aussian or binary data.
\newblock \emph{J. Mach. Learn. Res.}, 9:\penalty0 485--516,
  2008{\natexlab{a}}.
\newblock ISSN 1532-4435.

\bibitem[Banerjee et~al.(2008{\natexlab{b}})Banerjee, Gelfand, Finley, and
  Sang]{banerjee2008}
Sudipto Banerjee, Alan~E. Gelfand, Andrew~O. Finley, and Huiyan Sang.
\newblock Gaussian predictive process models for large spatial data sets.
\newblock \emph{J. R. Stat. Soc. Ser. B Stat. Methodol.}, 70\penalty0
  (4):\penalty0 825--848, 2008{\natexlab{b}}.
\newblock ISSN 1369-7412.
\newblock \doi{10.1111/j.1467-9868.2008.00663.x}.
\newblock URL \url{https://doi.org/10.1111/j.1467-9868.2008.00663.x}.

\bibitem[Barab\'{a}si and Albert(1999)]{BA}
Albert-L\'{a}szl\'{o} Barab\'{a}si and R\'{e}ka Albert.
\newblock Emergence of scaling in random networks.
\newblock \emph{Science}, 286\penalty0 (5439):\penalty0 509--512, 1999.
\newblock ISSN 0036-8075.
\newblock \doi{10.1126/science.286.5439.509}.
\newblock URL \url{https://doi.org/10.1126/science.286.5439.509}.

\bibitem[Bien and Tibshirani(2011)]{bien}
Jacob Bien and Robert~J. Tibshirani.
\newblock Sparse estimation of a covariance matrix.
\newblock \emph{Biometrika}, 98\penalty0 (4):\penalty0 807--820, 2011.
\newblock ISSN 0006-3444.
\newblock \doi{10.1093/biomet/asr054}.
\newblock URL \url{https://doi.org/10.1093/biomet/asr054}.

\bibitem[Bolin and Lindgren(2011)]{bolin}
David Bolin and Finn Lindgren.
\newblock Spatial models generated by nested stochastic partial differential
  equations, with an application to global ozone mapping.
\newblock \emph{Ann. Appl. Stat.}, 5\penalty0 (1):\penalty0 523--550, 2011.
\newblock ISSN 1932-6157.
\newblock \doi{10.1214/10-AOAS383}.
\newblock URL \url{https://doi.org/10.1214/10-AOAS383}.

\bibitem[Boyd and Vandenberghe(2004)]{boydCVX}
Stephen Boyd and Lieven Vandenberghe.
\newblock Convex optimization, 2004.

\bibitem[Cressie and Johannesson(2008)]{cressie2008}
Noel Cressie and Gardar Johannesson.
\newblock Fixed rank kriging for very large spatial data sets.
\newblock \emph{J. R. Stat. Soc. Ser. B Stat. Methodol.}, 70\penalty0
  (1):\penalty0 209--226, 2008.
\newblock ISSN 1369-7412.
\newblock \doi{10.1111/j.1467-9868.2007.00633.x}.
\newblock URL \url{https://doi.org/10.1111/j.1467-9868.2007.00633.x}.

\bibitem[Cressie and Wikle(2011)]{wikle}
Noel Cressie and Chris Wikle.
\newblock \emph{Statistics for Spatio-Temporal Data}.
\newblock CourseSmart Series. Wiley, 2011.
\newblock ISBN 9780471692744.
\newblock URL \url{https://books.google.com/books?id=-kOC6D0DiNYC}.

\bibitem[{Davanloo Tajbakhsh} et~al.(2014){Davanloo Tajbakhsh}, {Serhat Aybat},
  and {Del Castillo}]{davanloo}
Sam {Davanloo Tajbakhsh}, Necdet {Serhat Aybat}, and Enrique {Del Castillo}.
\newblock Sparse precision matrix selection for fitting gaussian random field
  models to large data sets.
\newblock \emph{ArXiv e-prints}, May 2014.

\bibitem[Dinh~Tao and Le~Thi(1997)]{dc}
Pham Dinh~Tao and Hoai~An Le~Thi.
\newblock Convex analysis approach to {DC} programming: Theory, algorithm and
  applications.
\newblock \emph{Acta Mathematica Vietnamica}, 22, 01 1997.

\bibitem[Fattahi et~al.(2019)Fattahi, Zhang, and Sojoudi]{lineartime2019}
Salar Fattahi, Richard~Y. Zhang, and Somayeh Sojoudi.
\newblock Linear-time algorithm for learning large-scale sparse graphical
  models.
\newblock \emph{IEEE Access}, pages 1--1, 2019.
\newblock ISSN 2169-3536.
\newblock \doi{10.1109/ACCESS.2018.2890583}.

\bibitem[Friedman et~al.(2007)Friedman, Hastie, H\"{o}fling, and
  Tibshirani]{pathwise}
Jerome Friedman, Trevor Hastie, Holger H\"{o}fling, and Robert Tibshirani.
\newblock Pathwise coordinate optimization.
\newblock \emph{Ann. Appl. Stat.}, 1\penalty0 (2):\penalty0 302--332, 2007.
\newblock ISSN 1932-6157.
\newblock \doi{10.1214/07-AOAS131}.
\newblock URL \url{https://doi.org/10.1214/07-AOAS131}.

\bibitem[Friedman et~al.(2008)Friedman, Hastie, and
  Tibshirani]{friedman_sparse_2008}
Jerome Friedman, Trevor Hastie, and Robert Tibshirani.
\newblock Sparse inverse covariance estimation with the graphical lasso.
\newblock \emph{Biostatistics}, 9\penalty0 (3):\penalty0 432--441, July 2008.
\newblock ISSN 1465-4644.
\newblock \doi{10.1093/biostatistics/kxm045}.
\newblock URL \url{https://www.ncbi.nlm.nih.gov/pmc/articles/PMC3019769/}.

\bibitem[Fuentes(2002)]{fuentes2002}
Montserrat Fuentes.
\newblock Spectral methods for nonstationary spatial processes.
\newblock \emph{Biometrika}, 89\penalty0 (1):\penalty0 197--210, 2002.
\newblock ISSN 0006-3444.
\newblock \doi{10.1093/biomet/89.1.197}.
\newblock URL \url{https://doi.org/10.1093/biomet/89.1.197}.

\bibitem[Gneiting and Raftery(2007)]{CRPS}
Tilmann Gneiting and Adrian~E. Raftery.
\newblock Strictly proper scoring rules, prediction, and estimation.
\newblock \emph{J. Amer. Statist. Assoc.}, 102\penalty0 (477):\penalty0
  359--378, 2007.
\newblock ISSN 0162-1459.
\newblock \doi{10.1198/016214506000001437}.
\newblock URL \url{https://doi.org/10.1198/016214506000001437}.

\bibitem[Guhaniyogi and Banerjee(2018)]{guhaniyogi2017}
Rajarshi Guhaniyogi and Sudipto Banerjee.
\newblock Meta-kriging: scalable {B}ayesian modeling and inference for massive
  spatial datasets.
\newblock \emph{Technometrics}, 60\penalty0 (4):\penalty0 430--444, 2018.
\newblock ISSN 0040-1706.
\newblock \doi{10.1080/00401706.2018.1437474}.
\newblock URL \url{https://doi.org/10.1080/00401706.2018.1437474}.

\bibitem[Hsieh et~al.(2014)Hsieh, Sustik, Dhillon, and Ravikumar]{QUIC}
Cho-Jui Hsieh, M\'{a}ty\'{a}s~A. Sustik, Inderjit~S. Dhillon, and Pradeep
  Ravikumar.
\newblock Q{UIC}: quadratic approximation for sparse inverse covariance
  estimation.
\newblock \emph{J. Mach. Learn. Res.}, 15:\penalty0 2911--2947, 2014.
\newblock ISSN 1532-4435.

\bibitem[Hunter and Lange(2004)]{mm}
David~R. Hunter and Kenneth Lange.
\newblock A tutorial on {MM} algorithms.
\newblock \emph{The American Statistician}, 58:\penalty0 30--37, 02 2004.

\bibitem[Katzfuss(2017)]{katzfuss}
Matthias Katzfuss.
\newblock A multi-resolution approximation for massive spatial datasets.
\newblock \emph{J. Amer. Statist. Assoc.}, 112\penalty0 (517):\penalty0
  201--214, 2017.
\newblock ISSN 0162-1459.
\newblock \doi{10.1080/01621459.2015.1123632}.
\newblock URL \url{https://doi.org/10.1080/01621459.2015.1123632}.

\bibitem[Legates and Willmott(1990)]{Legates1990}
David~R. Legates and Cort~J. Willmott.
\newblock Mean seasonal and spatial variability in global surface air
  temperature.
\newblock \emph{Theoretical and Applied Climatology}, 41\penalty0 (1):\penalty0
  11--21, Mar 1990.
\newblock ISSN 1434-4483.
\newblock \doi{10.1007/BF00866198}.
\newblock URL \url{https://doi.org/10.1007/BF00866198}.

\bibitem[Lindgren et~al.(2011)Lindgren, Rue, and Lindstr\"{o}m]{lindgren2011}
Finn Lindgren, H{\aa}vard Rue, and Johan Lindstr\"{o}m.
\newblock An explicit link between {G}aussian fields and {G}aussian {M}arkov
  random fields: the stochastic partial differential equation approach.
\newblock \emph{J. R. Stat. Soc. Ser. B Stat. Methodol.}, 73\penalty0
  (4):\penalty0 423--498, 2011.
\newblock ISSN 1369-7412.
\newblock \doi{10.1111/j.1467-9868.2011.00777.x}.
\newblock URL \url{https://doi.org/10.1111/j.1467-9868.2011.00777.x}.
\newblock With discussion and a reply by the authors.

\bibitem[Matsuda and Yajima(2009)]{matsuda2009}
Yasumasa Matsuda and Yoshihiro Yajima.
\newblock Fourier analysis of irregularly spaced data on {$\Bbb{R}^d$}.
\newblock \emph{J. R. Stat. Soc. Ser. B Stat. Methodol.}, 71\penalty0
  (1):\penalty0 191--217, 2009.
\newblock ISSN 1369-7412.
\newblock \doi{10.1111/j.1467-9868.2008.00685.x}.
\newblock URL \url{https://doi.org/10.1111/j.1467-9868.2008.00685.x}.

\bibitem[Nychka(2000)]{nychka2000}
Douglas Nychka.
\newblock {Spatial-process estimates as smoothers}.
\newblock In M.~G. Schimek, editor, \emph{Smoothing and Regression: Approaches,
  Computation and Application}, pages 393--424. New York: Wiley, 2000.

\bibitem[Nychka et~al.(2002)Nychka, Wikle, and Royle]{nychka2002}
Douglas Nychka, Christopher Wikle, and J.~Andrew Royle.
\newblock Multiresolution models for nonstationary spatial covariance
  functions.
\newblock \emph{Stat. Model.}, 2\penalty0 (4):\penalty0 315--331, 2002.
\newblock ISSN 1471-082X.
\newblock \doi{10.1191/1471082x02st037oa}.
\newblock URL \url{https://doi.org/10.1191/1471082x02st037oa}.

\bibitem[Nychka et~al.(2015)Nychka, Bandyopadhyay, Hammerling, Lindgren, and
  Sain]{nychka2015}
Douglas Nychka, Soutir Bandyopadhyay, Dorit Hammerling, Finn Lindgren, and
  Stephan Sain.
\newblock A multiresolution {G}aussian process model for the analysis of large
  spatial datasets.
\newblock \emph{J. Comput. Graph. Statist.}, 24\penalty0 (2):\penalty0
  579--599, 2015.
\newblock ISSN 1061-8600.
\newblock \doi{10.1080/10618600.2014.914946}.
\newblock URL \url{https://doi.org/10.1080/10618600.2014.914946}.

\bibitem[Oyler et~al.(2015)Oyler, Ballantyne, Jencso, Sweet, and
  Running]{oyler2015}
Jared~Wesley Oyler, Ashley Ballantyne, Kelsey Jencso, Michael Sweet, and
  Steven~W. Running.
\newblock {Creating a topoclimatic daily air temperature dataset for the
  conterminous United States using homogenized station data and remotely sensed
  land skin temperature}.
\newblock \emph{International Journal of Climatology}, 35:\penalty0 2258--2279,
  2015.

\bibitem[Roweis and Ghahramani(1999)]{unifying}
Sam Roweis and Zoubin Ghahramani.
\newblock A unifying review of linear gaussian models.
\newblock \emph{Neural computation}, 11:\penalty0 305--45, 03 1999.

\bibitem[Rue and Held(2005)]{rue2005}
H{\aa}vard Rue and Leonhard Held.
\newblock \emph{Gaussian Markov Random Fields: Theory And Applications
  (Monographs on Statistics and Applied Probability)}.
\newblock Chapman \& Hall/CRC, 2005.
\newblock ISBN 1584884320.

\bibitem[Seber(2007)]{matrixhandbook}
George A.~F. Seber.
\newblock \emph{A Matrix Handbook for Statisticians}.
\newblock Wiley-Interscience, New York, NY, USA, 1st edition, 2007.
\newblock ISBN 0471748692, 9780471748694.

\bibitem[Tipping(2001)]{rvm}
Michael~E. Tipping.
\newblock Sparse {B}ayesian learning and the relevance vector machine.
\newblock \emph{J. Mach. Learn. Res.}, 1\penalty0 (3):\penalty0 211--244, 2001.
\newblock ISSN 1532-4435.
\newblock \doi{10.1162/15324430152748236}.
\newblock URL \url{https://doi.org/10.1162/15324430152748236}.

\bibitem[Yuan and Lin(2007)]{yuanlin}
Ming Yuan and Yi~Lin.
\newblock Model selection and estimation in the {G}aussian graphical model.
\newblock \emph{Biometrika}, 94\penalty0 (1):\penalty0 19--35, 2007.
\newblock ISSN 0006-3444.
\newblock \doi{10.1093/biomet/asm018}.
\newblock URL \url{https://doi.org/10.1093/biomet/asm018}.

\bibitem[Zhao et~al.(2015)Zhao, Li, Liu, Roeder, Lafferty, and Wasserman]{huge}
Tuo Zhao, Xingguo Li, Han Liu, Kathryn Roeder, John Lafferty, and Larry
  Wasserman.
\newblock \emph{huge: High-Dimensional Undirected Graph Estimation}, 2015.
\newblock URL \url{https://CRAN.R-project.org/package=huge}.
\newblock R package version 1.2.7.

\end{thebibliography}

\section{Appendix}

We start the Appendix with the proof of Proposition \ref{prop:likelihood}.  
\begin{proof}[Proof of Proposition \ref{prop:likelihood}]
  For matrices $A, U, C$ and $V$ of appropriate size, the Sherman Woodbury 
  Morrison formula is 
  \[
    \left(A+UCV \right)^{-1} = A^{-1} - 
    A^{-1}U \left(C^{-1}+VA^{-1}U \right)^{-1} VA^{-1}
  \]
  and the matrix determinant lemma is 
  \[
    \det(A + U C V) = \text{det}(C^{-1} + V A^{-1} U) \ \text{det}(C) \ \text{det}(A).
  \]
  In our case, these two equations read
  \begin{equation} \label{ourSMW}
    \left( \Phi Q^{-1} \Phi^{\mathrm{T}} + \tau^2 I_n \right)^{-1} 
    = \frac{1}{\tau^2} I_n 
    - \frac{1}{\tau^4} \Phi \left(  Q + \frac{1}{\tau^2} \Phi^{\mathrm{T}} 
    \Phi \right) ^{-1} \Phi^{\mathrm{T}}
  \end{equation}
  and 
  \begin{equation} \label{ourMDL}
    \det( \Phi Q^{-1} \Phi^{\mathrm{T}}+\tau^2 I_n ) 
    = \det \left(Q + \frac{ 1}{ \tau^2}\Phi^{\mathrm{T}} \Phi \right) 
    \det(Q^{-1}) \det(\tau^2 I_n).
  \end{equation}
  Combining (\ref{ourSMW}) with linearity and the cyclic property of trace gives
  \[
    \text{tr}(S (\Phi Q^{-1} \Phi^{\mathrm{T}} + \tau^2 I_n)^{-1} ) 
    = \frac{1}{\tau^2}\text{tr}(S) - 
  \text{tr} \left(   \frac{1}{\tau^4} \Phi^{\mathrm{T}} S \Phi 
    \left(  Q + \frac{1}{\tau^2} \Phi^{\mathrm{T}} \Phi \right) ^{-1}  \right),
  \]
  and taking the logarithm of (\ref{ourMDL}) immediately yields 
  \[
    \log \det  ( \Phi Q^{-1} \Phi^{\mathrm{T}} + \tau^2 I_n) 
    = \log \det \left(Q + \frac{ 1}{ \tau^2}\Phi^{\mathrm{T}} \Phi \right) 
    - \log \det Q  + n \log \tau^2. \hfill \qedhere
  \]
\end{proof}

\subsection*{Convexity, Gradients, and Hessians}

The penalized likelihood in Proposition \ref{prop:likelihood} was as follows:
\[
 \log \det  \left(Q + \frac{ 1 }{ \tau^2} \Phi^T \Phi \right)  - \log \det (Q)  - \text{tr} \left( \dfrac{1}{\tau^4} \Phi^{\mathrm{T}}  S \Phi   \left( Q + \frac{1}{\tau^2} \Phi^{\mathrm{T}} \Phi \right)^{-1}  \right) + \| \lambda \circ Q\|_1
\]
Let us explain the classifications of convexity and concavity stated in the opening paragraph of Section \ref{algorithmsection}. The penalty function $Q \mapsto \| \Lambda \circ Q \|_1$ is trivially convex on $Q \succeq 0$.
An explicit proof of concavity for $Q \mapsto \log \det Q$ on $Q \succeq 0$ is given on page 74 of \citep{boydCVX}. The convexity of $\text{tr}( A Q^{-1})$ on $Q \succeq 0$ for an arbitrary positive semidefinite matrix $A$ can be shown in a similar fashion, as the authors suggest in their Exercise 3.18(a). Composition with an affine mapping preserves both convexity and concavity, so $Q \mapsto \log \det \left( Q + \frac{1}{\tau^2} \Phi^{\mathrm{T}} \Phi \right)$ is concave on $Q \succeq 0$ and $Q \mapsto \text{tr} \left(A \left( Q + \frac{1}{\tau^2} \Phi^{\mathrm{T}} \Phi \right) \right)^{-1}$ is convex on $Q \succeq 0$. 

Below we report the gradient and Hessian matrices of the first three terms in this penalized likelihood, with $\otimes$ indicating the Kronecker product of two matrices. Let $W=Q^{-1}$ and $M = \left( Q + \frac{1}{\tau^2} \Phi^{\mathrm{T}} \Phi \right)^{-1}$ for shorthand.
\begin{center}
$\begin{array}{c||c c}
\hline \hline
  & \text{Gradient} & \text{Hessian} \\  \hline
\log \det  \left(Q + \frac{ 1 }{ \tau^2} \Phi^T \Phi \right) &M & - (M \otimes M) \\[5pt]
- \log \det (Q) & - W & W \otimes W \\ 
- \dfrac{1}{\tau^4} \text{tr} \left(\Phi^{\mathrm{T}}  S \Phi  M  \right)
& \frac{1}{\tau^4} M \Phi^{\mathrm{T}} S \Phi M & - \frac{1}{\tau^4} \left( M \Phi^{\mathrm{T}} S \Phi M \otimes M \right) - \frac{1}{\tau^4} \left( M \otimes M \Phi^{\mathrm{T}} S \Phi M \right)  
\end{array}$
\end{center}
Our claims of convexity and concavity in the above paragraphs can be further verified with the following fact: if $\{\lambda_i\}$ and $\{ \mu_i\}$ are the eigenvalues of $A$ and $B$, then $A \otimes B$ has eigenvalues $\{\lambda_i \mu_j\}$ (11.5, \citep{matrixhandbook}).

\subsection*{DC algorithm in \texttt{R}}

We provide an algorithm for (\ref{QUICproblem}) in \texttt{R} assuming the \texttt{QUIC} package is installed. Mathematical inputs for the algorithm are the nugget effect $\tau^2$ and the $\ell \times \ell$ matrices $\Phi^{\mathrm{T}} \Phi$, $\Phi^{\mathrm{T}} S \Phi$, and $\Lambda$.

\begin{center}
\scriptsize
\begin{Verbatim}[frame=single]
QUIC_DCsolve <- function(lambda,guess,tau_sq,Phi_Phi,Phi_S_Phi,tol,max_iterations,max_runtime_seconds){ 
     norm <- 1
     counter <- 1
     Phi_Phi_over_nugget <- Phi_Phi/tau_sq
     Phi_S_Phi_over_nuggetsq <- Phi_S_Phi / (tau_sq^2)
     start.time <- Sys.time()
     elapsed <- 0
     while (norm >= tol && counter <= max_iterations && elapsed <= max_runtime_seconds) {
          M <- chol2inv(chol(guess + Phi_Phi_over_nugget))
          S_star <- (diag(dim(guess)[1]) + M %*% Phi_S_Phi_over_nuggetsq) %*% M
          new_guess <- QUIC(S_star,rho=lambda,msg=0)$X
          current.time <- Sys.time()
          elapsed <- difftime(current.time,start.time,units="secs")
          norm <- norm(new_guess - guess,type="F")/norm(guess,type="F")
          cat("Iteration ", counter, ". Relative error: ", norm, ". Time elapsed: ", elapsed,"\n", sep="")
          guess <- new_guess
          counter <- counter+1    
      }
      if(counter > max_iterations) {
          cat("Maximum Iterations Reached. Relative error: ", sprintf("%10f",norm), sep ="")
          cat("\n")
      } else if(elapsed > max_runtime_seconds) { 
          cat("Maximum run time reached. Relative error: ", sprintf("%10f",norm), sep ="")
          cat("\n")	
      } else {
          cat("Convergence. Relative error: ", sprintf("%10f",norm), sep ="")
          cat("\n")
          }	
     guess
}
\end{Verbatim}
\normalsize
\end{center}

\subsection*{Alternative DC Algorithm}

In the body of this paper, we chose to group the summands of the likelihood (\ref{Qlikelihood}) individually so that $-\log \det (Q)$ is the only convex term and everything else is concave and hence linearized into the trace term under the DC framework, ultimately giving rise to the graphical lasso subproblem. It is interesting to note that $A \otimes A \succeq B \otimes B \iff A \succeq B$ (11.7, \citep{matrixhandbook}), which, in light of the Hessian calculations shown in the previous section, implies that the entire log determinant contribution
\[
\log \det  \left(Q + \frac{ 1 }{ \tau^2}\Phi^T \Phi \right)- \log \det (Q)
\]
is convex rather than just $- \log \det (Q)$. The DC scheme (\ref{DC}) in this case would read
\begin{equation} \label{newDC}
  Q_{j+1} = \argmin_{Q \succeq 0}   \left( \log \det  \left(Q + \frac{1 }{ \tau^2} \Phi^T \Phi \right)  - \log \det Q+\text{tr}    \left(  \nabla g(Q_j) Q\right) + \|\Lambda \circ Q\|_1 \right)
\end{equation}
where now the concave function $g(Q) = - \text{tr} \left(  \dfrac{1}{\tau^4} \Phi^{\mathrm{T}}  S \Phi  \left(Q + \frac{ 1}{ \tau^2} \Phi^T \Phi  \right)^{-1}  \right)$ is simply the trace term from the likelihood (\ref{Qlikelihood}). The hope is that linearizing less of the original function at each step may produce a different (better) stationary point, but (\ref{newDC}) is an unstudied convex problem and we cannot directly rely on the graphical lasso framework. Fortunately, several concepts from the \texttt{QUIC} algorithm extend directly to the new problem, and a few straightforward modifications to the \texttt{QUIC} source code produces a fast method for (\ref{newDC}).

On small, synthetic datasets, the estimated $\hat Q$ under the new algorithm (\ref{newDC}) and the graphical lasso DC algorithm (\ref{QUICproblem}) were indistinguishable, and the latter was faster to reach the strict tolerance $\varepsilon \ll 0.01$. In the case of the real datasets discussed in Section \ref{data}, the dimension $\ell$ is large enough so that there is a substantial difference in runtime between the two methods, but again the resulting estimates are essentially the same under a looser tolerance $\varepsilon=0.01$. We therefore chose to motivate our problem in this paper under the graphical lasso DC framework (\ref{QUICproblem}), a choice which feels natural given the interpretation of the graphical lasso as estimating the precision matrix in an undirected Gaussian graphical model. Nevertheless, it is interesting to note that that these two DC schemes seem to converge to the same estimate $\hat{Q}$ despite the fact that the original likelihood is not necessarily convex.

\end{document}